\documentclass[11pt]{article}
\usepackage[utf8]{inputenc}
\usepackage[colorlinks=true
]{hyperref}
\hypersetup{urlcolor=blue, citecolor=red, linkcolor=blue}
\usepackage{mleftright}
\usepackage{stmaryrd}
\SetSymbolFont{stmry}{bold}{U}{stmry}{m}{n}
\usepackage{microtype}

\usepackage{bbm}
\usepackage{colortbl}

\usepackage{graphicx}
\usepackage[active]{srcltx}

\usepackage{amsmath,amssymb}

\numberwithin{equation}{section}

\usepackage{physics}
\usepackage{bm}
\usepackage{enumerate}
\usepackage{amsthm}
\usepackage[all,2cell]{xy}
\usepackage{mathrsfs}
\usepackage{mathtools}
\mathtoolsset{showonlyrefs}

\usepackage{tikz-cd}
\usetikzlibrary{cd}

\usepackage{xcolor}
\DeclareUnicodeCharacter{0301}{*************************************}

\hypersetup{urlcolor=blue, citecolor=red, linkcolor=blue}

\usepackage[a4paper, left=25mm, right=25mm, top=30mm, bottom=25mm]{geometry}

\usepackage{marginnote}

\usepackage{imakeidx}
\makeindex[intoc]

\usepackage[nottoc]{tocbibind}
\usepackage[colorinlistoftodos]{todonotes} 

\theoremstyle{plain}
\newtheorem{theorem}{Theorem}[section]
\newtheorem{proposition}[theorem]{Proposition}

\newtheorem{lemma}[theorem]{Lemma}

\theoremstyle{definition}
\newtheorem{definition}[theorem]{Definition}
\newtheorem{remark}[theorem]{Remark}
\newtheorem{example}[theorem]{Example}

\AtBeginEnvironment{remark}{%
  \pushQED{\qed}%
}
\AtEndEnvironment{remark}{\popQED\endexample}

\AtBeginEnvironment{example}{%
  \pushQED{\qed}%
}
\AtEndEnvironment{example}{\popQED\endexample}

\newcommand\restr[2]{{
  \left.\kern-\nulldelimiterspace #1 \right|_{#2} 
}}

\newcommand{\R}{\mathbb{R}}

\renewcommand{\d}{\mathrm{d}}

\newcommand{\df}{\Omega}
\newcommand{\Cinfty}{\mathscr{C}^\infty}
\newcommand{\T}{\mathrm{T}}
\newcommand{\Tan}{\mathrm{T}}
\newcommand{\cT}{\mathrm{T}^\ast}

\newcommand*{\inn}[1]{\iota_{#1}}

\newcommand{\Lie}{\mathscr{L}}
\newcommand{\vf}{\mathfrak{X}}

\newcommand*{\bd}{\overline{\mathrm{d}}}

\newcommand{\bfX}{\mathbf{X}}
\newcommand{\bfY}{\mathbf{Y}}

\newcommand{\parder}[2]{\frac{\partial #1}{\partial #2}}

\newcommand{\tparder}[2]{\partial #1/\partial #2}

\DeclareMathOperator{\Diff}{Diff}
\DeclareMathOperator{\rk}{rk}

\DeclareMathAlphabet{\mathpzc}{OT1}{pzc}{m}{it}
\def\d{\mathrm{d}}

\newcommand{\dfn}[1]{\textsl{\textbf{{#1}}}}

\usepackage{xcolor}
\usepackage{fancyhdr}

\begin{document}


\vspace{3em}

{\huge\sffamily\raggedright Symmetries and Noether's theorem\\[.5ex]for multicontact field theories
}
\vspace{2em}

{\large\raggedright
    \today
}

\vspace{3em}

{\Large\raggedright\sffamily
    Xavier Rivas
}\vspace{1mm}\newline
{\raggedright
    Department of Computer Engineering and Mathematics, Universitat Rovira i Virgili\\
    Avinguda Països Catalans 26, 43007 Tarragona, Spain\\
    e-mail: \href{mailto:xavier.rivas@urv.cat}{xavier.rivas@urv.cat} --- {\sc orcid}: \href{https://orcid.org/0000-0002-4175-5157}{0000-0002-4175-5157}
}

\medskip

{\Large\raggedright\sffamily
    Narciso Román-Roy
}\vspace{1mm}\newline
{\raggedright
    Department of Mathematics, Universitat Politècnica de Catalunya\\
    C. Jordi Girona 31, 08034 Barcelona, Spain\\
    e-mail: \href{mailto:narciso.roman@upc.edu}{narciso.roman@upc.edu } --- {\sc orcid}: \href{https://orcid.org/0000-0003-3663-9861}{0000-0003-3663-9861}
}

\medskip

{\Large\raggedright\sffamily
    Bartosz M. Zawora
}\vspace{1mm}\newline
{\raggedright
    Department of Mathematical Methods in Physics, University of Warsaw, \\ ul. Pasteura 5, 02-093, Warszawa, Poland.\\
    e-mail: \href{mailto:b.zawora@uw.edu.pl}{b.zawora@uw.edu.pl} --- {\sc orcid}: \href{https://orcid.org/0000-0003-4160-1411}{0000-0003-4160-1411}
}

\vspace{3em}

{\large\bf\raggedright
    Abstract
}\vspace{1mm}\newline
{\raggedright
A geometric framework, called {\sl multicontact geometry}, has recently been developed to study action-dependent field theories. In this work, we use this framework to analyze symmetries in action-dependent Lagrangian and Hamiltonian field theories, as well as their associated dissipation laws. Specifically, we establish the definitions of conserved and dissipated quantities, define the general symmetries of the field equations and the geometric structure, and examine their properties. The latter ones, referred to as {\sl Noether symmetries}, lead to the formulation of a version of Noether's Theorem in this setting, which associates each of these symmetries with the corresponding dissipated quantity and the resulting conservation law.
}

\vspace{3em}

\medskip


\medskip






\noindent {\bf Authors' contributions:} All authors contributed to the study conception and design. The manuscript was written and revised by all authors. All authors read and approved the final version.
\medskip

\noindent {\bf Competing Interests:} The authors have no competing interests to declare. 

\newpage

{\setcounter{tocdepth}{2}
\def\baselinestretch{1}
\small
\def\addvspace#1{\vskip 1pt}
\parskip 0pt plus 0.1mm
\tableofcontents
}

\pagestyle{fancy}

\fancyhead[L]{Noether's theorem for multicontact field theories --}    
\fancyhead[C]{}                  
\fancyhead[R]{X. Rivas, N. Román-Roy and B.M. Zawora}       

\fancyfoot[L]{}     
\fancyfoot[C]{\thepage}                  
\fancyfoot[R]{}            

\setlength{\headheight}{17pt}

\renewcommand{\headrulewidth}{0.1pt}  
\renewcommand{\footrulewidth}{0pt}    

\renewcommand{\headrule}{%
    \vspace{3pt}                
    \hrule width\headwidth height 0.4pt 
    \vspace{0pt}                
}

\setlength{\headsep}{30pt}  

\section{Introduction}

Action-dependent theories in physics are extensions of standard dynamical and classical field theories, which exhibit properties that make them particularly interesting.
They are characterized by the fact that their Lagrangian or Hamiltonian functions incorporate additional variables that are related to the action
and lead to the emergence of new terms in the field or dynamical equations. 
These extra terms have been used to model dissipative behavior; although the applications of these theories extend far beyond that scope.
Related to this aspect is the non-conservative nature of these kinds of systems; however, despite this, the deviation from conservation is clearly regulated.
In fact, in this framework, Noether's Theorem links symmetries to dissipated quantities, which, although not conserved, follow a well-established behavior.

In recent years, there has been growing interest in developing a geometrical framework to study mechanical systems that exhibit dissipation.
This framework is {\sl contact geometry} \cite{BH_16,Gei_08,Kho_13},
which has been applied to describe dissipative Hamiltonian \cite{BCT_17,LL_19,GG_22} and  Lagrangian systems \cite{CCM_18,LL_19a,GGMRR_20a},
and has also been applied to other various fields, such as
thermodynamics, quantum mechanics, circuit theory, or control theory, 
\cite{Bra_18,CCM_18,Kho_13}

The same dissipative phenomena arise in classical field theories. Drawing inspiration from contact geometry, a geometric framework called {\sl multicontact geometry} has been recently introduced, providing a systematic and rigorous way to set the Lagrangian and Hamiltonian formalisms for action-dependent field theories. 
This subject has been recently developed in \cite{LGMRR_23} (although a previous, different approach was also presented in \cite{Vit_15}).
This structure and its application to describe these theories is based on the {\sl multisymplectic formulation} of classical field theory.
The literature on this subject is vast; a few among the many works available are as follows: for instance, 
see \cite{AA_80,Gar_74,Sau_89} for the Lagrangian formalism,
\cite{CCI_91,FPR_05,GS_73,HK_02} for the Hamiltonian formalism,
\cite{LMM_96a} for the singular case, and
\cite{GIMM_98,Rom_09,RW_19} as general references
(see also \cite{LSV_15,GMS_97,Kan_98} for other geometric frameworks; the so-called {\sl polysymplectic} and {\sl $k$-symplectic} formulations). Additionally, an alternative but related approach to action-dependent field theories involves the so-called {\sl k-contact structures} \cite{LRS_24, GGMRR_20, GGMRR_21}.

As it is well-known, symmetries play a crucial role in studying dynamical systems and field theories. Exploring symmetries often involves analyzing ordinary and partial differential equations, as they encode conservation laws and conserved quantities that reveal fundamental properties of physical systems. 
As a cornerstone in this field, Emmy Noether's groundbreaking work \cite{Noe_18} established a deep connection between certain kinds of symmetries and conservation principles (see also \cite{Kos_11}, for a modern perspective). 
Initially, symmetry in physical systems was understood as the invariance of the equations of motion under specific transformations in phase space, but, in modern geometric formulations of classical mechanics and field theories, symmetries are typically defined through the invariance of fundamental geometric structures, ensuring the preservation of the equations of motion as a natural consequence.

In particular, in multisymplectic field theories, the study of symmetries and their associated conserved quantities and conservation laws has been extensively developed within the geometric framework, where symmetries that preserve the (pre)multisymplectic structure are known as Noether or Cartan symmetries, 
This geometric formulation of Noether symmetries allows to easily prove Noether's Theorem.
In particular, see, for instance,
\cite{FS_12,GIMM_98,RWZ_20} as general references
and, specifically, 
\cite{AA_78,BBLSV_15,LMS_04,GR_23} for the Lagrangian case, and
\cite{EMR_99,GR_24} for the Hamiltonian one).

The goal of this work is to extend the results on symmetries, conservation laws, and Noether's Theorem from the classical multisymplectic field theory to the case of action-dependent field theories within the multicontact framework, both for the Lagrangian and Hamiltonian formulations.
In this context, symmetries are associated with dissipation laws.
Specifically, we aim to define the concepts of conserved and dissipated quantities in this framework, as well as the notion of symmetry for both the field equations and the multicontact structure;
and the properties of all these kinds of symmetries.
The last ones are referred to as {\sl strong Noether symmetries},
and the ultimate objective is establishing the corresponding version of Noether's Theorem, linking them to their corresponding dissipated quantities. A similar partial analysis of these subjects but in the $k$-contact framework of action-dependent field theories is found in  \cite{GGMRR_20, GGMRR_21}.

The organization of the paper is as follows:
Section \ref{prelim} reviews some preliminary concepts on multicontact geometry and the multicontact formulation of Lagrangian and Hamiltonian action-dependent field theories.
Section \ref{cdfsym} contains the main statements and results of the work.
First, we define conserved and dissipative forms and explain how they lead to conservation and dissipation laws. 
Second, we introduce different types of symmetries for multicontact Lagrangian systems; namely, {\sl generalized}, {\sl Noether}, and {\sl strong Noether symmetries}. 
After defining the canonical lifts to the phase space of the theory, the so-called {\sl natural symmetries} are distinguished, and the relationships between all these symmetries are studied, along with their properties. 
The section continues with the statement and proof of {\sl Noether's Theorem} and a discussion on how all these results extend to multicontact Hamiltonian systems. Finally, in Section \ref{applications}, an example is analyzed: a symmetry of the damped wave equation. Appendix \ref{append} is also included to describe multivector fields and their properties.  

All the manifolds are real, second-countable, connected, and of class $\Cinfty$, and the mappings are assumed to be smooth.
Sum over crossed repeated indices is understood.

\section{Preliminaries}
\label{prelim}

This section reviews the basics on multicontact geometry and the multicontact Lagrangian and Hamiltonian formulations of classical field theories. See \cite{LGMRR_23} for more details.

\subsection{Review on multicontact geometry}

Consider a manifold $P$ with $\dim{P}=m+N$, where $N\geq m\geq 1$, endowed with two differential $m$-forms $\Theta,\omega\in\df^m(P)$ of constant ranks. Let $\mathcal{D}\subset\Tan P$ be a regular distribution and let $\Gamma(\mathcal{D})$ be a $\Cinfty(P)$-module of sections. For each $k\in\mathbb{N}$, define
\[
\mathcal{A}^k(\mathcal{D}):=\big\{ \alpha\in\df^k(P)\mid
\inn{Z}\alpha=0 \,,\ \forall Z\in\Gamma({\cal D})\big\} \,,
\]
namely the set of $k$-forms on $P$ that vanish when evaluated on vector fields in $\Gamma(\mathcal{D})$. At a point $\mathrm{p}\in P$, the corresponding pointwise version is
$
\mathcal{A}_{\mathrm{p}}^k(\mathcal{D}) := \big\{ \alpha\in\bigwedge^k\Tan^*_\mathrm{p}P\mid
\inn{v}\alpha=0\,,\ \forall v\in{\cal D}_\mathrm{p}\big\}$.
This definition of ${\cal A}^k({\cal D})$ can equivalently be rewritten as
\[
{\cal A}^k({\cal D})=\big\{ \alpha\in\df^k(P) \mid 
\Gamma(\cal D)\subset\ker\alpha \big\} \,,
\]
where $\ker\alpha =\{ Z\in\vf(P)\mid \inn{Z}\alpha=0\}$ is the {\it one-kernel} of a differential form $\alpha\in\df^k(P)$, with $k\geq 1$.
\begin{definition}
\label{def:reeb_dist}
The \dfn{Reeb distribution} associated with the pair $(\Theta,\omega)$ is the distribution
${\cal D}^{\mathfrak{R}}\subset\Tan P$ defined at each point $\mathrm{p}\in P$, as
\[
{\cal D}_\mathrm{p}^{\mathfrak{R}}=\big\{ v\in\ker\omega_\mathrm{p}\ \mid\ \inn{v}\dd\Theta_\mathrm{p}\in {\cal A}^{m}_\mathrm{p}(\ker\omega)\big\} \,,
\]
with $\displaystyle{\cal D}^{\mathfrak{R}}=\coprod_{\mathrm{p}\in P}{\cal D}_\mathrm{p}^{\mathfrak{R}}$.
The set of sections of the Reeb distribution is denoted by $\mathfrak{R}:=\Gamma({\cal D}^{\mathfrak{R}})$,
and its elements $R\in\mathfrak{R}$ are called \dfn{Reeb vector fields}.
Assuming that $\ker \omega$ is a distribution of constant rank, then
\begin{equation}
\label{Reebdef}
\mathfrak{R}=\big\{ R\in\Gamma(\ker\omega)\,|\, \inn{R}\dd\Theta\in {\cal A}^{m}(\ker\omega)\big\} \,.
\end{equation}
\end{definition}

Observe that $\ker\omega\cap\ker\dd\Theta\subseteq{\cal D}^\mathfrak{R}$.
Furthermore, if $\omega\in\Omega^m(P)$ is a closed differential form with constant rank, then $\mathfrak{R}$ forms an involutive distribution.

\begin{definition}
\label{premulticontactdef}
A pair $(\Theta,\omega)$ is a \dfn{premulticontact structure} on $P$ if $\omega\in\Omega^m(P)$ is closed and, for $0\leq k\leq N-m$,
we have the following:
\begin{enumerate}[{\rm (1)}]
\item\label{prekeromega}
$\rk\ker\omega=N$.
\item\label{prerankReeb}
$\rk{\cal D}^{\mathfrak{R}}=m+k$.
\item\label{prerankcar}
$\rk\left(\ker\omega\cap\ker\Theta\cap\ker\d\Theta\right)=k$.
\item \label{preReebComp}
${\cal A}^{m-1}(\ker\omega)=\{\inn{R}\Theta\mid R\in \mathfrak{R}\}$,
\end{enumerate}
Then, the triple $(P,\Theta,\omega)$ is said to be
a \dfn{premulticontact manifold}
and $\Theta\in\Omega^m(P)$ is called a \dfn{premulticontact form} on $P$. 
The distribution
$\mathcal{C}\equiv\ker\omega\cap\ker\Theta\cap\ker\d\Theta$
is the \dfn{characteristic distribution} of $(P,\Theta,\omega)$.
If $k=0$, the pair $(\Theta,\omega)$ is a \dfn{multicontact structure}, the triple
$(P,\Theta,\omega)$ is a \dfn{multicontact manifold},
and $\Theta\in\Omega^m(P)$ is a \dfn{multicontact form}.
\end{definition}

An essential property of these structures is the following.

\begin{proposition}
\label{sigma}
Let $(P,\Theta,\omega)$ be a (pre)multicontact manifold, then
there exists a unique $\sigma_{\Theta}\in\df^1(P)$, called the \dfn{dissipation form}, satisfying
\[
\sigma_{\Theta}\wedge\inn{R}\Theta=\inn{R}\dd\Theta \,,\qquad \text{for every }\,R\in\mathfrak{R} \,.
\]
\end{proposition}

Then, using the dissipation form one can define the following.

\begin{definition}
\label{def:bard}
Let $\sigma_{\Theta}\in\df^1(P)$ be the dissipation form. Then, we define the operator associated with $\sigma_{\Theta}$ as follows
\begin{align*}
\bd:\df^k(P)&\longrightarrow\df^{k+1}(P)
\\
\beta&\longmapsto\bd\beta=\dd \beta+\sigma_{\Theta}\wedge\beta\,.
\end{align*}
\end{definition}

In field theories, the (pre)multisymplectic structures satisfy the following additional condition, which guarantees that the theory is variational (see \cite{GLMR_24}).

\begin{definition}
Let \((P,\Theta,\omega)\) be a (pre)multicontact manifold satisfying
\begin{equation}
\iota_X\iota_{Y}\Theta = 0 \,, \quad
\text{for every } X,Y\in\Gamma(\ker\omega) \,.
\end{equation}
Then $(P,\Theta,\omega)$ is a \dfn{variational (pre)multicontact manifold} and $(\Theta,\omega)$ is said to be a \dfn{variational (pre)multicontact structure}.
\end{definition}

\subsection{Multicontact Lagrangian systems}
\label{sect:mcs}

Let $\pi\colon E\to M$ be the {\sl configuration bundle} of a classical first-order field theory, 
where $M$ is an orientable $m$-dimensional manifold with volume form $\eta\in\df^m(M)$, and
let $J^1\pi\to E\to M$ be the corresponding first-order jet bundle.
If $\dim M=m$ and $\dim E=n+m$, then $\dim J^1\pi=nm+n+m$.
Natural coordinates in $J^1\pi$ adapted to the bundle structure
are $(x^\mu,y^a,y^a_\mu)$ ($\mu = 1,\ldots,m$; $a=1,\ldots,n$),
and are such that
$\eta=\d x^1\wedge\cdots\wedge\d x^m=:\d^mx$.

The multicontact Lagrangian formalism for action-dependent (or non-conservative) field theories takes place in the bundle, 
${\cal P}=J^1\pi\times_M\bigwedge^{m-1}\Tan^*M$,
where, $\bigwedge^{m-1}\Tan^*M$ denotes the bundle of $(m-1)$-forms on $M$.
Now the {\sl configuration bundle} is $\mathscr{E}\equiv E\times_M\bigwedge^{m-1}\Tan^*M\to M$.
Natural coordinates in ${\cal P}$ are $(x^\mu,y^a,y^a_\mu,s^\mu)$,
and thus $\dim\mathcal{P}=2m+n+nm$. Note that in these coordinates $\omega:=\tau^*\eta=\d^m x\in\Omega^m(\mathcal{P})$, where $\tau\colon \mathcal{P}\rightarrow M$ is the natural projection. We have the natural projections depicted in the following diagram:
\begin{equation}
\label{firstdiagram}
\xymatrix{
&\ &  \  &{\cal P}=J^1\pi\times_M\bigwedge^{m-1}\Tan^*M \ar[rrd]_{\tau_2}\ar[lld]^{\tau_1}\ar[ddd]_{\tau}\ &  \   &
\\
&J^1\pi \ar[ddrr]^{\bar{\pi}^1}\ar[d]^{\pi^1}\ & \ & \ & \ &\bigwedge^{m-1}\Tan^*M
 \ar[ddll]^{\kappa}
\\
&E\ar[drr]_{\pi}\ & \ & \ & \ &  
 \\
&\ & \ &M \ & \ & 
}
\end{equation}

A section $\bm{\psi}\colon M\rightarrow {\cal P}$ of $\tau$ is said to be a \dfn{holonomic section} in ${\cal P}$ if
the section $\psi:=\tau_1\circ\bm{\psi}\colon M\to J^1\pi$
is holonomic in $J^1\pi$; that is,
there is a section $\phi\colon M\to E$ of $\pi$ such that $\psi=j^1\phi$.
It is customary to write $\bm{\psi}=(\psi,s)=(j^1\phi,s)$, 
where $s\colon M\to\bigwedge^{m-1}\Tan^*M$  is a section of the projection $\kappa\colon 
\bigwedge^{m-1}\Tan^*M\to M$. We also say that $\bm{\psi}$ is the
\dfn{canonical prolongation} of the section 
$\bm\phi:=(\phi,s)\colon M\to\mathscr{E}\equiv E\times_M\bigwedge^{m-1}\Tan^*M$ to ${\cal P}$.

\begin{definition}
A multivector field
$\bfX\in\vf^m({\cal P})$ is a \dfn{holonomic $m$-multivector field}
or a \dfn{second-order partial differential equation} (\textsc{sopde}) in ${\cal P}$ if
it is $\tau$-transverse (see Appendix~\ref{append}), integrable, and its integral sections are holonomic on ${\cal P}$.
\end{definition}

The $\tau$-transversality condition can be taken to be $\inn{\bfX}\omega=1$, since if $\bfX$ satisfies $\inn{\bfX}\omega \neq 0$, there always exists a nonvanishing function $f$ such that $\bfX' = f\bfX$ satisfies $\inn{\bfX'}\omega=1$. Then,
the local expression of a locally decomposable {\sc sopde} in ${\cal P}$ reads
\[
\bfX = \bigwedge^m_{\mu=1}
\Big(\parder{}{x^\mu}+y^i_\mu\frac{\displaystyle\partial} {\displaystyle
\partial y^i}+F_{\mu\nu}^i\frac{\displaystyle\partial}{\displaystyle \partial y^i_\nu}+g^\nu_\mu\,\frac{\partial}{\partial s^\nu}\Big)\,,
\]
and its integral sections are solutions to the system of second-order partial differential equations:
\begin{equation}
    y^i_\mu=\parder{y^i}{x^\mu}\,,\qquad F^i_{\mu\nu}=\frac{\partial^2y^i}{\partial x^\mu \partial x^\nu}\,.
\end{equation}
Note that the functions $F_{\mu\nu}^i$ and $g_\mu^\nu$ are not completely arbitrary, since they need to satisfy some integrability condition (see Appendix \ref{append}).
Multivector fields in \({\cal P}\) that possess the above local expression given above but are not necessarily integrable are commonly referred to as {\sl semi-holonomic}.

Physical information in field theories is introduced through the so-called {\sl Lagrangian densities}.
A \dfn{Lagrangian density} is a differential form $\mathcal{L}\in\df^m({\cal P})$;
hence $\mathcal{L}=L\,\d^mx$, where $L\in\Cinfty({\cal P})$ is the
\dfn{Lagrangian function}.
Then, the \dfn{Lagrangian $m$-form} associated with $\mathcal{L}$
is defined using the canonical geometric elements which ${\cal P}$ is endowed with, and its coordinate expression reads
\begin{equation}
\label{thetacoor1}
    \Theta_{\mathcal{L}}=
    -\frac{\partial L}{\partial y^a_\mu}\d y^a\wedge\d^{m-1}x_\mu +\left(\frac{\partial L}{\partial y^a_\mu}y^a_\mu-L\right)\d^m x+\d s^\mu\wedge \d^{m-1}x_\mu\in\df^m({\cal P}) \,,
\end{equation}
where $\d^{m-1}x_\mu = \inn{\tparder{}{x^\mu}}\d^m x$. The local function
$\displaystyle E_\mathcal{L}=\frac{\partial L}{\partial y^a_\mu}y^a_\mu-L$
is called the \dfn{energy Lagrangian function} associated with $L\in\Cinfty(\mathcal{P})$.
A Lagrangian function $L\in\Cinfty({\cal P})$ is \dfn{regular} if the Hessian matrix
$\displaystyle 
\left(\frac{\partial^2L}{\partial y^a_\mu\partial y^b_\nu}\right)$
is regular everywhere; then $\Theta_\mathcal{L}$ is a \dfn{Lagrangian multicontact form} on ${\cal P}$
and the triple $({\cal P},\Theta_\mathcal{L},\omega)$ is said to be a \dfn{multicontact Lagrangian system}.
Otherwise, $L$ is a \dfn{singular} Lagrangian and, under certain additional conditions, $\Theta_\mathcal{L}$ is a {\sl premulticontact form} on ${\cal P}$
and $({\cal P},\Theta_\mathcal{L},\omega)$ is a \dfn{premulticontact Lagrangian system}.

The next step is to introduce the \dfn{dissipation form} which, in the Lagrangian formalism, has the coordinate expression of the form
\begin{equation}
\label{sigmaL}
\sigma_{\Theta_\mathcal{L}}=-\parder{L}{s^\mu}\,\d x^\mu \,.
\end{equation}
Then, finally, we construct the form,
$$
\overline\d\Theta_\mathcal{L} =\d\Theta_\mathcal{L}+\sigma_{\Theta_\mathcal{L}}\wedge\Theta_\mathcal{L}
 = \d\Theta_\mathcal{L}-\parder{L}{s^\nu}\d x^\nu\wedge\Theta_\mathcal{L}\,.
 $$

For a multicontact Lagrangian system $({\cal P},\Theta_\mathcal{L},\omega)$
the Lagrangian field equations are derived from the {\sl generalized Herglotz Variational Principle} \cite{GLMR_24},
and can be stated alternatively as:

\begin{definition}
\label{mconteqH}
\begin{enumerate}[{\rm(1)}]
\item  
The \dfn{multicontact Lagrangian equations for holonomic sections}
$\bm{\psi}\colon M\to{\cal P}$ are
\begin{equation}
\label{sect1L}
\bm{\psi}^*\Theta_{\cal L}= 0  \,,\qquad
\bm{\psi}^*\inn{X}\bd\Theta_{\cal L}= 0 \,, \qquad \text{for every }\ X\in\vf({\cal P}) \,.
\end{equation}
\item 
The \dfn{multicontact Lagrangian equations for 
holonomic multivector fields} $\bfX \in\vf^m({\cal P})$ are
\begin{equation}
\label{vfH}
\inn{\mathbf{X}}\Theta_{\cal L}=0 \,,\qquad \inn{\bfX}\bd\Theta_{\cal L}=0 \,.
\end{equation}
Recall that holonomic multivector fields are $\tau$-transverse.
Then, note that equations \eqref{vfH} and the $\tau$-transversality condition, $\inn{\bfX}\omega\not=0$, hold for every multivector field of the equivalence class $\{\bfX\}$ (that is, for every $\bfX'=f\bfX$, with $f$ nonvanishing; see the Appendix \ref{append}).
Then, the condition of $\tau$-transversality can be imposed by simply asking that $\inn{\bfX}\omega=1$.
\end{enumerate}
\end{definition}

In coordinates, for holonomic sections
$\displaystyle\bm{\psi}(x^\nu)=\Big(x^\mu,y^a(x^\nu),\parder{y^a}{x^\mu}(x^\nu),s^\mu(x^\nu)\Big)$,
equations \eqref{sect1L} read
\begin{align}
 \parder{s^\mu}{x^\mu}&=L\circ{\bm{\psi}} \,,
 \label{ELeqs2}
 \\
\label{ELeqs1}
\frac{\partial}{\partial x^\mu}
\left(\frac{\displaystyle\partial L}{\partial
y^a_\mu}\circ{\bm{\psi}}\right)&=
\left(\frac{\partial L}{\partial y^a}+
\displaystyle\frac{\partial L}{\partial s^\mu}\displaystyle\frac{\partial L}{\partial y^a_\mu}\right)\circ{\bm{\psi}} \,,
\end{align}
which are precisely the equations for the integral sections of {\sc sopde} multivector fields solutions to the equations \eqref{vfH}.
Equations \eqref{ELeqs1}
are called the {\sl Herglotz--Euler--Lagrange field equations}.

\begin{theorem}
\label{solutions}
Let $({\cal P},\Theta_\mathcal{L},\omega)$ be a (pre)multicontact Lagrangian system.
If $\bfX\in\mathfrak{X}^m({\cal P})$ is a holonomic multivector field solution to \eqref{vfH} (that is, a {\sc sopde}), its integral sections are the solutions to the multicontact Euler--Lagrange field equations for holonomic sections \eqref{sect1L}.

In addition, if the Lagrangian system is regular then:
\begin{enumerate}[\rm (1)]
\item
The multicontact Lagrangian field equations
for multivector fields \eqref{vfH} 
have solution on ${\cal P}$,
which is not unique if $m>1$.
\item
Multivector fields $\bfX\in \mathfrak{X}^m({\cal P})$ which are solutions to equations \eqref{vfH} are semi-holonomic
(and hence holonomic, if they are integrable).
\end{enumerate}
\end{theorem}

\begin{remark}
If the Lagrangian $L$ is not regular and $({\cal P},\Theta_\mathcal{L},\omega)$
is a premulticontact system; in general, 
the field equations \eqref{vfH} have no solutions everywhere on ${\cal P}$
and, when they do and are integrable, they are not necessarily {\sc sopde}s.
Hence, the requirement to be {\sc sopde} is imposed as an additional condition.
In the best situations, {\sc sopde} solutions exist only on a submanifold ${\cal P}_f\hookrightarrow{\cal P}$ which is obtained by applying a suitable constraint algorithm.    
\end{remark}

In this paper, we only consider the case of regular Lagrangians;
that is, multicontact Lagrangian systems.
In the singular (premulticontact) case, the diffeomorphisms considered must leave the submanifold ${\cal P}_f$ invariant and all the vector fields and multivector fields must be tangent to ${\cal P}_f$; then, all results must be restricted to ${\cal P}_f$.
Nevertheless, as (semi-)holonomy is to be imposed as an additional condition, for singular Lagrangian systems,
this problem requires further research.

\subsection{Multicontact Hamiltonian systems (regular case)}
\label{mhf}

Let ${\cal M}\pi\equiv\bigwedge_2^m\Tan^*E$ be the bundle of $m$-forms on
$E$ vanishing by contraction with two $\pi$-vertical vector fields. It is endowed with natural coordinates $(x^\nu,y^a,p^\nu_a,p)$
adapted to the bundle structure ${\cal M}\pi\to E\to M$, and such that
$\eta=\d^mx$; so $\dim\, {\cal M}\pi=nm+n+m+1$.
Then, consider
$J^{1*}\pi\equiv{\cal M}\pi/\bigwedge^m_1\Tan^*E$
(where $\bigwedge^m_1\Tan^*E$ is the bundle of $\pi$-semibasic $m$-forms on
$E$);
whose natural coordinates are $(x^\mu,y^a,p_a^\mu)$, and so 
$\dim\,J^{1*}\pi=nm+n+m$.

For the Hamiltonian formulation of action-dependent first-order classical field theories, consider the bundles 
\[
\widetilde{\cal P}={\cal M}\pi\times_M\bigwedge\nolimits^{m-1}\Tan^* M
\,, \qquad
{\cal P}^* =J^{1*}\pi\times_M\bigwedge\nolimits^{m-1}\Tan^*M \,,
\]
which have natural coordinates $(x^\mu, y^i,p_i^\mu,p,s^\mu)$ and $(x^\mu,y^a,p_a^\mu,s^\mu)$, respectively.
We have the natural projections depicted in the following diagram:
$$
\xymatrix{
& &\widetilde{\cal P}={\cal M}\pi\times_M\bigwedge^{m-1}\Tan^* M
\ar[rddd]^{\widetilde\tau_2} \ar[llddd]_{\widetilde{\chi}} \ar[dd]_{\widetilde{\mathfrak{p}}} 
 \ar[lldd]_{\widetilde{\tau}_1} 
\\ & &
\\
{\cal M}\pi\ 
\ar[d]_{\mathfrak{p}}& &{\cal P}^*=J^{1*}\pi\times_M\bigwedge^{m-1}\Tan^* M
\ar@/_1pc/[uu]_{{\bf h}}
\ar[rd]_{\overline\tau_2}\ar[lld]^{\overline{\tau}_1}\ar[ddd]_{\overline\tau}
\\
J^{1*}\pi
\ar[ddrr]^{\overline{\pi}^1_M}\ar[d]_{\overline{\pi}^1}& & &\ar[ddl]_{\kappa}\bigwedge^{m-1}\Tan^* M
\\
E\ar[drr]_{\pi}
 \\
& \ &M \ & \ & 
}
$$
Since the bundles ${\cal M}\pi$ and  $\bigwedge^{m-1}(\Tan^* M)$ are bundles of forms, they have canonical structures, their ``tautological forms''
$\widetilde\Theta\in\df^m({\cal M}\pi)$ 
(called the \dfn{Liouville form} of ${\cal M}\pi$) and
$\theta\in \df^{m-1}(\bigwedge^{m-1}(\Tan^* M))$,
whose local expressions are
$$
\widetilde\Theta=p_a^\mu\d y^a\wedge\d^{m-1}x_\mu+p\,\d^m x
\,,\qquad
\theta=s^\mu\,\d^{m-1}x_\mu \,.
$$
Now, let ${\bf h}\colon {\cal P}^*\to\widetilde{\cal P}$ be a section of $\widetilde{\mathfrak{p}}$.
It is locally determined by a function $H\in\Cinfty(U)$, $U\subset{\cal P}^*$,
such that ${\bf h}(x^\mu,y^a,p^\mu_a)=(x^\mu,y^a,p^\mu_a,p=-H(x^\nu,y^b,p^\nu_b))$.
The elements ${\bf h}$ and $H$ are called
a \dfn{Hamiltonian section} and its associated \dfn{Hamiltonian function}, respectively.
Then, we define,
\begin{align}
\label{thetaHcoor}
\Theta_{\cal H} &= -(\widetilde{\tau}_1\circ{\bf h})^*\widetilde\Theta+\d(\overline\tau_2^*\theta)\nonumber\\
&=
-p_i^\mu\d y^i\wedge\d^{m-1}x_\mu+H\,\d^m x+\d s^\mu\wedge \d^{m-1}x_\mu\in\df^m({\cal P}^*) \,. 
\end{align}
The form 
$\Theta_{\cal H}$ is a variational multicontact form and a triple $({\cal P}^*,\Theta_{\cal H},\omega)$
is said to be a \dfn{multicontact Hamiltonian system}.
In this case, the {\sl dissipation form} is expressed as
\begin{equation}
\sigma_{\cal H}=\parder{H}{s^\mu}\,\d x^\mu\,.
\end{equation}

For multicontact Hamiltonian systems, the field equations
can be stated alternatively as:
\begin{enumerate}[{\rm(1)}]
\item  
The \dfn{multicontact Hamilton--de Donder--Weyl equations for sections}
$\bm{\psi}\colon M\to{\cal P}^*$ are
\begin{equation}
\label{sect1H}
\bm{\psi}^*\Theta_{\cal H}= 0  \,, \qquad
\bm{\psi}^*\inn{Y}\bd \Theta_{\cal H}= 0 \,, \qquad \text{for every }\ Y\in\vf({\cal }P^*) \, .
\end{equation}
\item 
The \dfn{multicontact Hamilton--de Donder--Weyl equations for $\overline\tau$-transverse and locally decomposable multivector fields} ${\bf X}_{\cal H}\in\vf^m({\cal P}^*)$ are
\begin{equation}
\label{vfH2}
\inn{\bfX_{\cal H}}\Theta_{\cal H}=0 \,, \qquad \inn{\bfX_{\cal H}}\bd\Theta_{\cal H}=0 \,.
\end{equation}
Equations \eqref{vfH2} and the $\overline\tau$-transversality condition hold for every multivector field of the equivalence class $\{\bfX\}$,
and the transversality condition can be imposed by asking $\inn{\bfX}\omega=1$.
\end{enumerate}
In natural coordinates, for a $\overline\tau$-transverse, locally decomposable multivector field ${\bf X}_{\cal H}\in~\vf^m({\cal P}^*)$, we have
\begin{equation}
\label{Hammv}
{\bf X}_{\cal H}=
\bigwedge_{\mu=0}^{m-1}\Big(\parder{}{x^\mu}+ (X_{\cal H})^a_\mu\frac{\partial}{\partial y^a}+
(X_{\cal H})_{\mu a}^\nu\frac{\partial}{\partial p_a^\nu}+(X_{\cal H})_\mu^\nu\parder{}{s^\nu}\Big) \,,
\end{equation} 
if it is a solution to equations \eqref{vfH2},
taking into account the local expression \eqref{thetaHcoor}
these field equations lead to
\begin{equation}
\label{coor2}
(X_{\cal H})^a_\mu=\frac{\partial H}{\partial p^\mu_a} \,,\qquad
(X_{\cal H})_{\mu a}^\mu= 
-\left(\frac{\partial H}{\partial y^a}+ p_a^\mu\,\frac{\partial H}{\partial s^\mu}\right)\,,\qquad (X_{\cal H})_\mu^\mu = 
p_a^\mu\,\frac{\partial H}{\partial p^\mu_a}-H \,,   
\end{equation}
along with a last group of equations which are identities when the above ones are taken into account.
Then, the integral sections $\bm{\psi}(x^\nu)=(x^\mu,y^a(x^\nu),p^\mu_a(x^\nu),s^\mu(x^\nu))$
of all the integrable solutions ${\bf X}_{\cal H}$, 
are solution to equations \eqref{sect1H} which read
\begin{equation}
\frac{\partial y^a}{\partial x^\mu}= \frac{\partial H}{\partial p^\mu_a}\circ\bm{\psi} \,,\qquad
\frac{\partial p^\mu_a}{\partial x^\mu} = 
-\left(\frac{\partial H}{\partial y^a}+ p_a^\mu\,\frac{\partial H}{\partial s^\mu}\right)\circ\bm{\psi}\,,\qquad 
\frac{\partial s^\mu}{\partial x^\mu} = \left(p_a^\mu\,\frac{\partial H}{\partial p^\mu_a}-H\right)\circ\bm{\psi}\,.
\label{coor1}
\end{equation}
These equation are called the \dfn{Herglotz--Hamilton--de Donder--Weyl equations}
for action-dependent classical field theories and are compatible in ${\cal P}^*$.

\section{Conservation and dissipation forms and symmetries}
\label{cdfsym}

This section introduces the dissipative quantities and proves the analogue of Noether's theorem for Lagrangian and Hamiltonian field theories. The definitions and results presented within this section are natural generalizations from the multisymplectic to the multicontact setting \cite{EMR_99, GIMM_98, GR_23, GR_24}.

\subsection{Conservation and dissipation laws}

Consider a multicontact Lagrangian system \((\mathcal{P},\Theta_{\mathcal{L}},\omega)\). Let
\[
\mathfrak{X}^m(\Theta_{\mathcal{L}}):=\{\mathbf{X}\in \mathfrak{X}^m(\mathcal{P})\,\mid\,\iota_{\mathbf{X}}\Theta_{\mathcal{L}}=0,\quad \iota_{\mathbf{X}}\bd\Theta_{\mathcal{L}}=0\}\,.
\]
Additionally, let
\[
\mathfrak{X}^m(\Theta_{\mathcal{L}},\omega):=\{\mathbf{X}\in \mathfrak{X}^m(\Theta_{\mathcal{L}})\,\mid\,\iota_{\mathbf{X}}\omega\neq 0 \}\,.
\]
Without loss of generality, one can assume that, in the above definition, \(\iota_{\mathbf{X}}\omega=1\), by the arguments presented in Section \ref{prelim}. Within this work we are mainly focused on locally decomposable multivector fields, hence the set of all locally decomposable multivector fields belonging to \(\mathfrak{X}^m(\Theta_{\mathcal{L}},\omega)\) is denoted by \(\mathfrak{X}^m_d(\Theta_{\mathcal{L}},\omega)\);
so it is made of the solutions to the equations \eqref{vfH}. 
Moreover, let \(\mathfrak{X}^m_{{\rm I},d}(\Theta_{\mathcal{L}},\omega)\) be the subset of all the multivector fields in \(\mathfrak{X}^m_d(\Theta_{\mathcal{L}},\omega)\) that are integrable. Then, we have
\[
\mathfrak{X}^m_{{\rm I},d}(\Theta_{\mathcal{L}},\omega)\subseteq\mathfrak{X}^m_d(\Theta_{\mathcal{L}},\omega)\subseteq\mathfrak{X}^m(\Theta_{\mathcal{L}},\omega)\subseteq \mathfrak{X}^m(\Theta_{\mathcal{L}})\,.
\]

\begin{definition}
Let $(\mathcal{P},\Theta_{\cal L},\omega)$ be a multicontact Lagrangian system.
\begin{enumerate}
\item A \dfn{conserved form} of the system is a differential form $\bm\xi\in\df^{m-1}({\cal P})$ such that, for every $\bfX \in\vf_d^m(\Theta_\mathcal{L},\omega)$,
\[
\Lie_\bfX \bm\xi=(-1)^{m+1}\iota_\bfX \d\bm\xi=0\,.
\]
\item A \dfn{dissipative form} of the system 
is a differential form $\bm\xi\in\df^{m-1}({\cal P})$ which satisfies that,
for every $\bfX \in\vf_d^m(\Theta_\mathcal{L},\omega)$,
\[
    \inn{\mathbf{X}}\bd\bm\xi=0\,.
\]
\end{enumerate}
\end{definition}

Note that if $\boldsymbol{\xi},\boldsymbol{\zeta}$ are conserved or dissipative forms, then $\boldsymbol{\xi}+\boldsymbol{\zeta}$ is also a conserved or dissipative form, respectively. Note that $\bm \xi\in\Omega^{m-1}(\cal{P})$ can be conserved and dissipative at the same time.

The first property of conserved forms is as follows:

\begin{proposition}
\label{Prop::ConQS}
Let $\bm \xi\in\df^{m-1}({\cal P})$ be a conserved form of \((\mathcal{P},\Theta_{\mathcal{L}},\omega)\) and let $\bfX \in\vf^m_{{\rm I},d}(\Theta_{\mathcal{L}},\omega)$. 
Then, $\bm\xi$ is closed on the integral submanifolds of $\bfX $. In other words, if $S\subset\cal P$ is an integral submanifold of $\bfX $ and $j_S\colon S\hookrightarrow{\cal P}$ is the natural embedding, then $j_S^*(\d\bm\xi)=0$.
\end{proposition}
\begin{proof}
Let $\bfX =X_1\wedge\dotsb\wedge X_m\in\vf^m_{{\rm I},d}(\Theta_{\mathcal{L}},\omega)$ with $X_1,\ldots ,X_m\in\vf({\cal P})$ independent vector fields tangent to the ($m$-dimensional) integral submanifold $S$. Then, since $\inn{\bfX}\d\bm\xi=0$, we have
\[
0=j_S^*[\inn{\bfX}\d\bm\xi]=
j_S^*[\inn{X_1\wedge\dotsb\wedge X_m}\d\bm\xi]=
j_S^*[\d\bm\xi (X_1,\ldots ,X_m)] \,,
\]
and as $X_1,\ldots ,X_m$ are arbitrary vector fields tangent to $S$, this implies that $j_S^*\d\bm\xi=0$.
\end{proof}

\begin{remark}
\label{conslaw}
Note that for every $\bm\xi\in\df^{m-1}({\cal P})$ and $\bfX\in\vf^m_{{\rm I},d}(\Theta_{\mathcal{L}},\omega)$,
if $\bm\psi\colon M\to{\cal P}$ is an integral section of $\bfX $, taking $\bm\psi^*\bm\xi\in\df^{m-1}(M)$,
then there is a unique vector field $X_{\bm\psi^*\bm\xi}\in\vf(M)$ such that $\inn{X_{\bm\psi^*\bm\xi}}\omega=\bm\psi^*\bm\xi$.
Recall that $\eta\in\df^m(M)$ is the volume form in $M$;
then, since the {\sl divergence} of $X_{\bm\psi^*\bm\xi}$ is the function 
${\rm div}X_{\bm\psi^*\bm\xi}\in \Cinfty(M)$ defined by, 
$$
\Lie_{X_{\bm\psi^*\bm\xi}}\eta= ({\rm div}X_{\bm\psi^*\bm\xi})\,\eta \,,
$$
it follows that $({\rm div}X_{\bm\psi^*\bm\xi})\,\eta=\d({\bm\psi^*\bm\xi})$. 
Therefore, by Proposition \ref{Prop::ConQS}, if $\bm\xi\in\Omega^{m-1}(\mathcal{P})$ is a conserved form, then $\d({\bm\psi^*\bm\xi})=0$, or, equivalently, ${\rm div}X_{\bm\psi^*\bm\xi}=0$. Thus, on every compact domain, $U\subset M$, {\sl Stokes theorem} leads to the {\sl conservation law},
\[
\int_{\partial U}{\bm\psi^*\bm\xi}=\int_U\d({\bm\psi^*\bm\xi})=\int_U \left({\rm div}X_{\bm\psi^*\bm\xi}\right)\,\eta=0 \,.
\]
The form $\bm\psi^*\bm\xi\in\Omega^{m-1}(M)$ is called the \dfn{conserved current}
associated with the conserved form $\bm\xi$.
In coordinates, $\bm\psi^*\bm\xi=f^\mu\,\d x^{m-1}x_\mu$, then
$\displaystyle X_{\bm\psi^*\bm\xi}=f^\mu\parder{}{x^\mu}$,
and the conservation law reads
\[
\displaystyle {\rm div}X_{\bm\psi^*\bm\xi}=\parder{f^\mu}{x^\mu}=0\,.
\]
\end{remark}

The analogous to Proposition \ref{Prop::ConQS} in dissipative setting is the following:

\begin{proposition}
\label{prop:dis}
    Let $\bm \xi\in\df^{m-1}({\cal P})$ be a dissipative form and let $\bfX \in\vf^m_{{\rm I},d}(\Theta_{\mathcal{L}},\omega)$. Then, $\bd \bm\xi$ vanish on integral submanifolds of $\bfX $; that is, $j_S^*(\bd\bm\xi)=\bd (j_S^*\bm\xi)=0$.
\end{proposition}
\begin{proof}
The proof is the same as in Proposition \ref{Prop::ConQS}.    
\end{proof}

\begin{remark}
\label{remdislaw}
Note that, as in Remark \ref{conslaw}, for every dissipative form $\bm\xi\in\df^{m-1}({\cal P})$ and $\bfX\in\vf^m_{{\rm I},d}(\Theta_{\mathcal{L}},\omega)$,
if $\bm\psi\colon M\to{\cal P}$ is an integral section of $\bfX $,
then there is a unique vector field $X_{\bm\psi^*\bm\xi}\in\vf(M)$ such that $\inn{X_{\bm\psi^*\bm\xi}}\eta=\bm\psi^*\bm\xi$. Therefore, as $\bm\xi$ is a dissipative form and as a consequence of Proposition \ref{prop:dis}, $\bd({\bm\psi^*\bm\xi})=0$ and hence,
by Definition \ref{def:bard},
$$
\left({\rm div}X_{\bm\psi^*\bm\xi}\right)\,\eta=\d({\bm\psi^*\bm\xi})=-(\bm\psi^*\sigma_{\Theta_\mathcal{L}})\wedge(\bm\psi^*\bm\xi) \,.
$$
Then, on every compact domain, $U\subset M$, Stokes theorem leads to the following {\sl dissipation law},
\[
\int_{\partial U}{\bm\psi^*\bm\xi}=\int_U\d({\bm\psi^*\bm\xi})=
\int_U \left({\rm div}X_{\bm\psi^*\bm\xi}\right)\,\eta=
-\int_U(\bm\psi^*\sigma_{\Theta_\mathcal{L}})\wedge(\bm\psi^*\bm\xi)\,.
\]
In coordinates, if $\bm\psi^*\bm\xi=f^\mu\,\d x^{m-1}x_\mu$ and
$\displaystyle X_{\bm\psi^*\bm\xi}=f^\mu\parder{}{x^\mu}$;
taking into account the local expression of $\sigma_{\cal L}\in \Omega^1(\cal{P})$ given by \eqref{sigmaL}, the dissipation law reads
\begin{equation}
\label{dislaweq}
\parder{f^\mu}{x^\mu}=\left(\parder{L}{s^\mu}\circ\bm\psi\right)\,f^\mu\,.
\end{equation}
\end{remark}

\subsection{Generalized symmetries}
\label{sect:gensym}

In physics, the idea of {\sl symmetry} is associated with transformations on the phase space of the theory, which transform solutions to the equations of the theory to other solutions.
The case in which such transformations are locally generated by vector fields through their flows is also distinguished, as this is the usual situation in physics.
In the context of multicontact field theories, the first definition in this sense is as follows:

\begin{definition}
\label{defsymm}
Let $({\cal P},\Theta_{\cal L},\omega)$ be a multicontact Lagrangian system.
\begin{enumerate}
\item A \dfn{generalized symmetry} of the system is a diffeomorphism $\Phi\colon {\cal P}\to{\cal P}$ 
such that $\Phi_*(\vf_d^m(\Theta_\mathcal{L},\omega))\subseteq\vf_d^m(\Theta_\mathcal{L},\omega)$.
\item 
An \dfn{infinitesimal generalized symmetry} of the system
is a vector field $Y\in\vf(\mathcal{P})$ whose local diffeomorphisms generated by the flow of $Y$
are local symmetries or, equivalently, 
$[Y,\bfX ]\in\vf_d^m(\Theta_\mathcal{L},\omega)$, for every $\bfX \in\vf_d^m(\Theta_\mathcal{L},\omega)$, where the bracket is the Schouten--Nijenhuis bracket (see Appendix \ref{append}).
\end{enumerate}
\end{definition}

Thus, (infinitesimal) generalized symmetries transform solutions to the field equations \eqref{vfH} to other solutions and, as a consequence of Theorem \ref{solutions}, all of them are semi-holonomic multivector fields.
Furthermore, (infinitesimal) generalized symmetries allow us to generate new conserved forms from another given one. In fact:

\begin{proposition}
\label{genconsforms}
Let $\bm\xi\in\df^{m-1}({\cal P})$ be a conserved form of the multicontact Lagrangian system
$({\cal P},\Theta_{\cal L},\omega)$. Then:
\begin{enumerate}
\item 
If $\Phi\in{\Diff}({\cal P})$ is a generalized symmetry, then $\Phi^*\bm\xi$ is also a conserved form.
\item
If $Y\in\vf({\cal P})$ is an infinitesimal generalized symmetry, then $\Lie_Y\bm\xi$ is also a conserved form.
\end{enumerate} 
\end{proposition}
\begin{proof}   
It is immediate. In fact, 
for every ${\bf X}\in\vf_d^m(\Theta_\mathcal{L},\omega)$, 
for the first item,
$$
\Lie_\bfX(\Phi^*\bm\xi)= \Phi^*\Lie_\bfX\bm\xi=0 \,.
$$
And, for the second item, bearing in mind \eqref{liebrac},
$$
\Lie_\bfX(\Lie_Y\bm\xi)=\Lie_{[\bfX,Y]}\bm\xi+\Lie_Y(\Lie_\bfX\bm\xi)=0 \,,
$$
since $[\bfX,Y]$ is an infinitesimal symmetry.
\end{proof}

\begin{remark}
\label{nongenerator}
For dissipative forms, this result does not hold, unless the generalized symmetry preserves $\sigma_{\Theta_\mathcal{L}}$.
In fact, if $\bm\xi$
is a dissipative form and $\Phi^*\sigma_{\Theta_\mathcal{L}}=\sigma_{\Theta_\mathcal{L}}$; then, for every ${\bf X}\in\vf_d^m(\Theta_\mathcal{L},\omega)$,we have
\begin{align}
\iota_\bfX\bd(\Phi^*\bm\xi)&=
\iota_\bfX\big(\d(\Phi^*\bm\xi)+\sigma_{\Theta_\mathcal{L}}\wedge(\Phi^*\bm\xi)\big)=
\iota_\bfX\big(\d(\Phi^*\bm\xi)+(\Phi^*\sigma_{\Theta_\mathcal{L}})\wedge(\Phi^*\bm\xi)\big)
\\ &=
\Phi^*\big(\iota_{((\Phi^{-1})_*\bfX)}(\d\bm\xi+\sigma_{\Theta_\mathcal{L}}\wedge\bm\xi)\big)=
\Phi^*(\iota_{\bfX'}\bd\bm\xi) =0
\,,
\end{align}
since $\bfX'=(\Phi^{-1})_*\bfX\in\vf_d^m(\Theta_\mathcal{L},\omega)$.
And the same happens for infinitesimal symmetries.
\end{remark}

In classical field theories, symmetries of interest correspond to diffeomorphisms and vector fields on ${\cal P}$
which restrict to diffeomorphisms and vector fields on $M$. Thus, we define:

\begin{definition}
Let $({\cal P},\Theta_{\cal L},\omega)$ be a multicontact Lagrangian system (see the diagram \eqref{firstdiagram}).
\begin{enumerate}
\item
A \dfn{restricted symmetry} of the system
is a generalized symmetry $\Phi\colon{\cal P}\to{\cal P}$ such that it
restricts to a diffeomorphism on $M$; that is, there exists $\varphi\in{\Diff}(M)$ satisfying that
\[
\tau\circ\Phi=\varphi\circ\tau \,.
\]
\item
An \dfn{infinitesimal restricted symmetry} of the system is a vector field $Y\in\vf({\cal P})$ whose local diffeomorphisms generated by the flow of $Y$ are local restricted symmetries or, what is equivalent, it is an infinitesimal generalized symmetry which is $\tau$-projectable; that is, there exists $Z\in\vf(M)$ satisfying that
\[
\tau_*Y=Z \,.
\]
\end{enumerate}
\end{definition}

This means that these (infinitesimal) symmetries preserve the fibration $\tau\colon{\cal P}\to M$. 
In coordinates, they read as
\begin{align}
&\Phi(x^\mu,y^a,y^a_\mu,s^\mu)= \big(\phi^\mu(x^\nu),\Phi^a(x^\nu,y^b,y^b_\nu,s^\nu),\Phi^a_\mu(x^\nu,y^b,y^b_\nu,s^\nu),\Psi^\mu(x^\nu,y^b,y^b_\nu,s^\nu)\big) \,,
\\
&Y=f^\mu(x^\nu)\parder{}{x^\mu}+F^a(x^\nu,y^b,y^b_\nu,s^\nu)\parder{}{y^a}+F^a_\mu(x^\nu,y^b,y^b_\nu,s^\nu)\parder{}{y^a_\mu}+g^\mu(x^\nu,y^b,y^b_\nu,s^\nu)\parder{}{s^\mu} \,.
\end{align}
Analogously, the so-called {\sl natural symmetries} are also particularly relevant in physics.
These are transformations induced in the phase space by transformations in the configuration space.
To introduce them in this multicontact context, it is necessary to define the concept of the {\sl canonical lift} to ${\cal P}$.
Hence, consider the manifold 
$\mathscr{E}\equiv E\times_M\bigwedge^{m-1}(\Tan^*M)$,
with the following projections
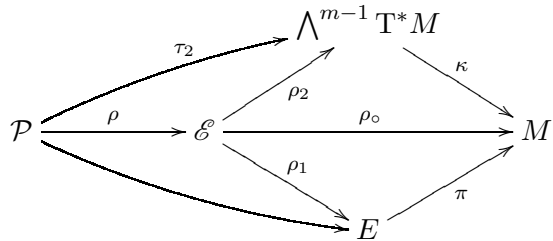
\begin{equation}
\xymatrix{
& & & \bigwedge^{m-1} \cT M \ar[dr]^{\kappa} \\
\mathcal{P} \ar[rr]^{\rho} \ar@/^.5pc/[urrr]^{\tau_2} \ar@/_.5pc/[drrr]_{} & & \mathscr{E} \ar[rr]^{\rho_\circ} \ar[ur]_{\rho_2} \ar[dr]^{\rho_1} & & M \\
& & & E \ar[ur]_{\pi}
}\nonumber
\end{equation}
Let $\Phi_\mathscr{E}\colon\mathscr{E}\to~\mathscr{E}$
be a diffeomorphism such that:
\begin{itemize}
\item[-] It restricts to a diffeomorphism $\Psi\colon\bigwedge^{m-1}(\Tan^*M)\to\bigwedge^{m-1}(\Tan^*M)$; that is, $\rho_2\circ\Phi_\mathscr{E}=\Psi\circ\rho_2$;
which means that $\Phi_\mathscr{E}$ preserves the fibration of $\rho_2$.
\item[-] It restricts to a diffeomorphism $\varphi\colon M\to M$; that is, $\rho_\circ\circ\Phi_\mathscr{E}=\varphi\circ\rho_\circ$;
which means that $\Phi_\mathscr{E}$ preserves the fibration of $\rho_\circ$.
\end{itemize}
In other words, the following diagrams commute
\begin{equation}
\xymatrix{
\mathscr{E}\ar[r]^{\Phi_\mathscr{E}}\ar[d]_{\rho_2}& \mathscr{E}\ar[d]_{\rho_2} & & \mathscr{E}\ar[d]_{\rho_\circ}\ar[r]^{\Phi_\mathscr{E}} & \mathscr{E}\ar[d]_{\rho_\circ}\\
\bigwedge^{m-1} \T^*M \ar[r]^{\Psi}& \bigwedge^{m-1}\T^*M & & M\ar[r]^\varphi & M \\
}\nonumber
\end{equation}
In coordinates,
\[
\Phi_\mathscr{E}(x^\mu,y^a,s^\mu)= \big(\phi^\mu(x^\nu),\Phi^a(x^\nu,y^b,s^\nu),\Psi^\mu(s^\nu)\big) \,.
\]
Then, as $\Phi_{\mathscr{E}}\colon\mathscr{E}\to\mathscr{E}$ preserves the fibration of $\rho_2$ and the fibers of the projection $\rho_2$ are identified with $E$,
we can take the restriction of $\Phi_\mathscr{E}$ to these fibers (which is obtained by `freezing' $s^\mu$),
and it is a diffeomorphism $\Phi_E\colon E\to E$.
Thus, we can write $\Phi_\mathscr{E}=(\Phi_E,\Psi)$.

Similarly, we can consider vector fields $Y_\mathscr{E}\in\vf(\mathscr{E})$ which are $\rho_2$-projectable
and $\rho_\circ$-projectable;
that is, there exist \(Y\in\vf(\bigwedge^{m-1}(\Tan^*M))\) and \(Z\in\vf(M)\) satisfying $(\rho_2)_*Y_\mathscr{E}=Y\in\vf(\bigwedge^{m-1}(\Tan^*M))$ and
$\rho_{\circ*}Y_\mathscr{E}=Z\in\vf(M)$.
In coordinates:
\begin{equation}
\label{vfY}
Y_\mathscr{E}=f^\mu(x^\nu)\parder{}{x^\mu}+F^a(x^\nu,y^b,s^\nu)\parder{}{y^a}+g^\mu(s^\nu)\parder{}{s^\mu} \,.
\end{equation}

Then, we define:

\begin{definition}
\begin{enumerate}
\item 
The \dfn{canonical lift of a diffeomorphism} $\Phi_\mathscr{E}\in\Diff(\mathscr{E})$ to ${\cal P}$
is the diffeomorphism $\Phi\colon{\cal P}\to
{\cal P}$ such that $\Phi:=(j^1\Phi_E,\Psi)$, where $j^1\Phi_E\colon J^1\pi\to J^1\pi$ is the canonical lift of $\Phi_E$ to $J^1\pi$
(see \cite{Sau_89}).
\item
The \dfn{canonical lift of a vector field} $Y_\mathscr{E}\in\vf(\mathscr{E})$ to ${\cal P}$ is the vector field $Y =  j^1Y_\mathscr{E}\in\vf({\cal P})$ whose local diffeomorphisms generated by the flow of $Y$ are the canonical lifts of the local diffeomorphisms generated by the flow of $Y_\mathscr{E}$.
\end{enumerate}
\end{definition}

In particular, in coordinates,
$$
j^1Y_\mathscr{E}=f^\mu\parder{}{x^\mu}+F^a\parder{}{y^a}-\left(\parder{F^a}{x^\mu}-
y^a_\nu\parder{f^\nu}{x^\mu}+y^b_\mu\parder{F^a}{y^b}\right)
\parder{}{y^a_\mu}+g^\mu\parder{}{s^\mu} \,,
$$
where $f^\mu = f^\mu(x^\nu)$ and $g^\mu = g^\mu(s^\nu)$. Recall that $\rho\colon{\cal P}\to\mathscr{E}$ is the natural projection, and then note that
\[
\rho\circ\Phi=\Phi_\mathscr{E}\circ\rho\,, \qquad
\rho_*(j^1Y_\mathscr{E})=Y_\mathscr{E} \,.
\]

\begin{definition}
Let $({\cal P},\Theta_{\cal L},\omega)$ be a multicontact Lagrangian system.
\begin{enumerate}
\item
A \dfn{natural symmetry} of the system
is a generalized symmetry $\Phi\colon{\cal P}\to{\cal P}$ such that it is the canonical lift of a diffeomorphism $\Phi_\mathscr{E}\in{\Diff}\,(\mathscr{E})$
which restricts to two diffeomorphisms on $M$ and $\bigwedge^{m-1}(\Tan^*M)$.
\item
An \dfn{infinitesimal natural symmetry} of the system
is an infinitesimal generalised symmetry $Y\in\vf({\cal P})$ such that it is the canonical lift of a $\rho_0$-projectable and $\rho_2$-projectable vector field $Y_\mathscr{E}\in\vf(\mathscr{E})$.
\end{enumerate}
\end{definition}

By definition, (infinitesimal) natural symmetries are restricted symmetries.

\subsection{Noether symmetries, Noether's theorem}

To establish {\sl Noether's theorem} it is necessary to consider a special type of symmetries:

 \begin{definition}
 \label{Noethersym}
Let $({\cal P},\Theta_{\cal L},\omega)$ be a multicontact Lagrangian system.
\begin{enumerate}
\item
A \dfn{Noether symmetry} of the system
is a diffeomorphism $\Phi\colon{\cal P}\to{\cal P}$ satisfying that:
\begin{enumerate}
\item[(i)]
$\Phi^*\Theta_{\mathcal{L}}=\Theta_{\mathcal{L}}$.
\end{enumerate}
A \dfn{strong Noether symmetry} is a Noether symmetry which preserves the whole multicontact structure; namely, it satisfies (i) and
\begin{enumerate}
\item[(ii)]
$\Phi^*\omega=\omega$.
\end{enumerate}
A (strong) Noether symmetry $\Phi\in{\Diff}({\cal P})$ is said to be a \dfn{restricted Noether symmetry} if it restricts to a diffeomorphism $\varphi\colon M\to M$; that is, such that $\tau\circ\Phi=\varphi\circ\tau$.

A (strong) Noether symmetry is said to be a \dfn{natural Noether symmetry} if $\Phi$ is the canonical lift of a diffeomorphism $\Phi_\mathscr{E}\in{\Diff}\,(\mathscr{E})$
which restricts to two diffeomorphisms on $M$ and $\bigwedge^{m-1}(\Tan^*M)$.
\item
An \dfn{infinitesimal Noether symmetry} of the system
is a vector field $Y\in\vf({\cal P})$
such that the local diffeomorphisms generated by the flow of $Y$ are local Noether symmetries;
or, equivalently:
\begin{enumerate}
\item[(i)]
$\Lie_Y\Theta_{\mathcal{L}}=0$.
\end{enumerate}
An \dfn{infinitesimal strong  Noether symmetry} is an infinitesimal Noether symmetry which satisfies (i) and 
\begin{enumerate}
\item[(ii)]
$\Lie_Y\omega=0$.
\end{enumerate}
An infinitesimal (strong) Noether symmetry $Y\in\mathfrak{X}(\cal P)$ is said to be a \dfn{restricted infinitesimal (strong) Noether symmetry} if $Y$ is a $\tau$-projectable vector field.

An infinitesimal (strong) Noether symmetry is said to be \dfn{natural infinitesimal (strong) Noether symmetry} if $Y\in\mathfrak{X}(\cal P)$ is the canonical lift of a $\rho_0$-projectable and $\rho_2$-projectable vector field $Y_\mathscr{E}\in\vf(\mathscr{E})$.
\end{enumerate}
\end{definition}

In particular, in coordinates, an infinitesimal natural Noether symmetry reads
$$
Y=f^\mu\parder{}{x^\mu}+F^a\parder{}{y^a}-\left(\parder{F^a}{x^\mu}-
y^a_\nu\parder{f^\nu}{x^\mu}+y^b_\mu\parder{F^a}{y^b}\right)
\parder{}{y^a_\mu}+g^\mu\parder{}{s^\mu} \,,
$$
where $g^\mu = g^\mu(s^\nu)$ and $f^\mu = f^\mu(x^\nu)$.

\begin{remark}
It is immediate to prove that for (strong) Noether symmetries $\Phi_1,\Phi_2\in \Diff(\mathcal{P})$, one obtains that their composition is also a (strong) Noether symmetry. Analogously, if $Y_1, Y_2\in\mathfrak{X}({\cal P})$ are infinitesimal (strong) Noether symmetries, it follows that $[Y_1, Y_2]$ is also an infinitesimal (strong) Noether symmetry.    
\end{remark}

The properties of these kinds of Noether symmetries are as follows:

\begin{lemma}
\label{Reebinv}
Let $({\cal P},\Theta_{\cal L},\omega)$ be a multicontact Lagrangian system, and
let $\mathfrak{R}_\mathcal{L}$ be the Reeb distribution related to $(\Theta_\mathcal{L},\omega)$.
\begin{enumerate}
\item
If $\Phi\in{\Diff}({\cal P})$ is a strong Noether symmetry, then,
for every $R\in\mathfrak{R}_\mathcal{L}$, we have that $\Phi_*R\in\mathfrak{R}_\mathcal{L}$.
\item 
If $Y\in\vf({\cal P})$ is an infinitesimal strong Noether symmetry, then,
for every $R\in\mathfrak{R}_\mathcal{L}$, we have that $[Y,R]\in\mathfrak{R}_\mathcal{L}$.
\end{enumerate}
\end{lemma}
\begin{proof}
Let $\Phi$ be a strong Noether symmetry. First, if $Z\in\ker\omega$, then $\Phi_*Z\in\ker\omega$.
In fact,
$$
\inn{Z}\omega=0 \quad \Longrightarrow\quad
0=(\Phi^{-1})^*(\iota_Z\omega)= \iota_{(\Phi_*Z)}\big((\Phi^{-1})^*\omega\big)=\iota_{(\Phi_*Z)}\omega \,.
$$
Second, if $\alpha\in{\cal A}^{m}(\ker\omega)$,
then $\Phi^*\alpha\in{\cal A}^{m}(\ker\omega)$.
In fact, for every $Z\in\ker\omega$,
$$
\inn{Z}\alpha=0 \quad \Longrightarrow\quad
0=\Phi^*(\inn{Z}\alpha)= \inn{((\Phi^{-1})_*Z)}=\inn{Z'}(\Phi^*\alpha) \,,
$$
since $Z'=(\Phi^{-1})_*Z$ is an arbitrary element of $\ker\omega$, as $Z$ is. 
Finally, for every $R\in\mathfrak{R}_\mathcal{L}$, there is $\alpha\in{\cal A}^{m}(\ker\omega)$ such that $\iota_R\d\Theta_\mathcal{L}=\alpha$; therefore,
$$
0=(\Phi^{-1})^*\big(\iota_R\d\Theta_\mathcal{L}-\alpha\big)=
\iota_{(\Phi_*R)}\big(\d(\Phi^{-1})^*\Theta_\mathcal{L}\big)-(\Phi^{-1})^*\alpha=
\iota_{(\Phi_*R)}\d\Theta_\mathcal{L}-\alpha' \,,
$$
since $\alpha'=\Phi^*\alpha\in{\cal A}^{m}(\ker\omega)$ and, hence, $\Phi_*R\in\mathfrak{R}_\mathcal{L}$.

The result for infinitesimal strong Noether symmetries is a consequence of the definition and the above result.
\end{proof}

\begin{lemma}
\label{Prop:Rsigma}
Let $({\cal P},\Theta_{\cal L},\omega)$ be a multicontact Lagrangian system.
\begin{enumerate}
\item
If $\Phi\in\Diff({\cal P})$ is a strong Noether symmetry, then
$\Phi^*\sigma_{\Theta_\mathcal{L}}=\sigma_{\Theta_\mathcal{L}}$.
\item
If $Y\in\vf({\cal P})$ is an infinitesimal strong Noether symmetry, then
$\Lie_Y\sigma_{\Theta_\mathcal{L}}=0$.
\end{enumerate}
\end{lemma}
\begin{proof}
Recall that $\sigma_{\Theta_\mathcal{L}}$ is the unique one-form satisfying $\sigma_{\Theta_\mathcal{L}}\wedge\iota_R\Theta_{\cal L}=\iota_R\d\Theta_{\cal L}$, for every $R\in \mathfrak{R}_\mathcal{L}$.
\begin{enumerate}
\item 
Let $\Phi\in{\Diff}({\cal P})$ be a strong Noether symmetry. Then,
\begin{align}
0 &=
\Phi^*\big(\sigma_{\Theta_\mathcal{L}}\wedge\iota_R\Theta_{\cal L}-\iota_R\d\Theta_{\cal L}\big)=
(\Phi^*\sigma_{\Theta_\mathcal{L}})\wedge\iota_{((\Phi^{-1})_*R)}(\Phi^*\Theta_{\cal L})-\iota_{((\Phi^{-1})_*R)}\d(\Phi^*\Theta_{\cal L})
\\ &=
(\Phi^*\sigma_{\Theta_\mathcal{L}})\wedge\iota_{R'}\Theta_{\cal L}-\iota_{R'}\d\Theta_{\cal L} \,,
\end{align}
and, as $R'=(\Phi^{-1})_*R$ is an arbitrary element of $\mathfrak{R}_\mathcal{L}$, and $\sigma_{\Theta_\mathcal{L}}$ is unique,
this implies that $\Phi^*\sigma_{\Theta_\mathcal{L}}=\sigma_{\Theta_\mathcal{L}}$.
\item
The result for infinitesimal strong Noether symmetries is a consequence of the definition and the
above item. However, we can prove it explicitly as follows:

Let $Y\in\mathfrak{X}(\cal P)$ be an infinitesimal strong Noether symmetry. By Lemma \ref{Reebinv}, it follows that $[Y,R]=R'$ takes values in $\mathfrak{R}_\mathcal{L}$. Then,
    \begin{align}
\Lie_Y(\sigma_{\Theta_\mathcal{L}}\wedge\iota_R\Theta_{\cal L})
&=\Lie_Y\sigma_{\Theta_\mathcal{L}}\wedge\Theta_{\cal L}+\sigma_{\Theta_\mathcal{L}}\wedge\Lie_Y\iota_R\Theta_{\cal L} \\ 
&=\Lie_Y\sigma_{\Theta_\mathcal{L}}\wedge\Theta_{\cal L}+\sigma_{\Theta_\mathcal{L}}\wedge\iota_{[Y,R]}\Theta_{\cal L}+\sigma_{\Theta_\mathcal{L}}\wedge\iota_R\Lie_Y\Theta_{\cal L}\\
&=\Lie_Y\sigma_{\Theta_\mathcal{L}}\wedge\Theta_{\cal L}+\sigma_{\Theta_\mathcal{L}}\wedge\iota_{[Y,R]}\Theta_{\cal L}=\Lie_Y\sigma_{\Theta_\mathcal{L}}\wedge\Theta_{\cal L}+\sigma_{\Theta_\mathcal{L}}\wedge\iota_{R'}\Theta_{\cal L}\,,   
\end{align}
and
\[
\Lie_Y\iota_R\d\Theta_{\cal L}=\iota_{[Y,R]}\d\Theta_{\cal L}+\iota_R\Lie_Y\d\Theta_{\cal L}=\iota_{R'}\d\Theta_{\cal L}\,.
\]
Thus, by definition of $\sigma_{\Theta_\mathcal{L}}$, it follows that
\begin{multline*}
\Lie_Y\sigma_{\Theta_\mathcal{L}}\wedge\iota_R\Theta_{\cal L}+\sigma_{\Theta_\mathcal{L}}\wedge\iota_{R'}\Theta_{\cal L}-\iota_{R'}\d\Theta_{\cal L}=
\\
\Lie_Y\sigma_{\Theta_\mathcal{L}}\wedge\iota_R\Theta_{\cal L}+\sigma_{\Theta_\mathcal{L}}\wedge\iota_{R}\Theta_{\cal L}-\iota_{R}\d\Theta_{\cal L}=\left(\Lie_Y\sigma_{\Theta_\mathcal{L}}+\sigma_{\Theta_{\cal L}} \right)\wedge\iota_R\Theta_{\cal L}-\iota_{R}\d\Theta_{\cal L}=0\,. 
\end{multline*}
Then, by uniqueness of $\sigma_{\Theta_\mathcal{L}}$ one gets that $\Lie_Y\sigma_{\Theta_\mathcal{L}}=0$, and the statement follows. 
\end{enumerate}
\end{proof}

Therefore, as a first consequence of this last lemma, we can establish the following relation:

\begin{proposition}
\label{N-gen}
Let $(\mathcal{P},\Theta_\mathcal{L},\omega)$ be a multicontact Lagrangian system.
\begin{enumerate}
\item 
Every strong Noether symmetry $\Phi\in{\Diff}(\cal P)$ of the system
is a generalized symmetry.
\item
Every infinitesimal strong Noether symmetry $Y\in\mathfrak{X}(\cal P)$ of the system
is an infinitesimal generalized symmetry.
\end{enumerate}
\end{proposition}
\begin{proof}
For every $\mathbf{X}\in\vf_d^m(\Theta_\mathcal{L},\omega)$, one has:
\begin{enumerate}
\item 
If $\Phi\in{\Diff}(\cal P)$ is a strong Noether symmetry, then
\begin{equation}
\iota_{(\Phi_*\bfX)}\Theta_\mathcal{L}=
\iota_{(\Phi_*\bfX)}((\Phi^{-1})^*\Theta_\mathcal{L})= 
(\Phi^{-1})^*(\iota_\bfX\Theta_\mathcal{L})=0 \,.
\end{equation}
Furthermore,
\begin{align}
\iota_{(\Phi_*\bfX)}\bd\Theta_\mathcal{L}&=
\iota_{(\Phi_*\bfX)}\big(\d\Theta_\mathcal{L}+\sigma_{\Theta_\mathcal{L}}\wedge\Theta_\mathcal{L}\big)
\\ &=
\iota_{(\Phi_*\bfX)}\big(\d((\Phi^{-1})^*\Theta_\mathcal{L})+((\Phi^{-1})^*\sigma_{\Theta_\mathcal{L}})\wedge((\Phi^{-1})^*\Theta_\mathcal{L})\big)
\\ &=
(\Phi^{-1})^*\big(\iota_\bfX(\d\Theta_\mathcal{L}+\sigma_{\Theta_\mathcal{L}}\wedge\Theta_\mathcal{L})\big)=
(\Phi^{-1})^*(\iota_\bfX\bd\Theta_\mathcal{L}) =0
\,.
\end{align}
Finally, the transversality condition is $\iota_{\bf X}\omega=f$, where $f$ is a nonvanishing function; then,
$$
0=(\Phi^{-1})^*(\iota_{\bf X}\omega-f)=
\iota_{(\Phi_*\bfX)}((\Phi^{-1})^*\omega)-(\Phi^{-1})^*f=
\iota_{(\Phi_*\bfX)}\omega-(\Phi^{-1})^*f \,.
$$
Since $f$ is an arbitrary nonvanishing function, we can chose it so that $(\Phi^{-1})^*f$ is also nonvanishing. Thus, the transversality condition holds for $\Phi_*\bfX$.
As a consequence of all this, $\Phi_*\bfX\in\vf_d^m(\Theta_\mathcal{L},\omega)$,
and hence $\Phi$ is a generalized symmetry.
\item
Although the result for infinitesimal strong Noether symmetries and infinitesimal generalized symmetries is a consequence of their definitions and of the above item,
it can be also obtained through an explicit calculation as follows:
First, recall that $\Lie_Y\Theta_{\cal L}=0$; then, for every $\mathbf{X}\in\vf_d^m(\Theta_\mathcal{L},\omega)$, one has
(see Proposition \ref{Prop::A1}):
\[
\inn{[Y,\mathbf{X}]}\Theta_\mathcal{L} = \Lie_{Y}\inn{\bfX}\Theta_\mathcal{L} - \inn{\bfX}\Lie_Y\Theta_\mathcal{L}=0\,.
\]
Second,
\begin{align*}
\inn{[Y,\mathbf{X}]}\bd\Theta_{\cal L} &= \inn{[Y,\mathbf{X}]}\d\Theta_\mathcal{L}+\inn{[Y,\mathbf{X}]}\left(\sigma_{\Theta_\mathcal{L}}\wedge\Theta_\mathcal{L}\right)
\\ &=
\Lie_Y\inn{\mathbf{X}}\d\Theta_\mathcal{L}-\inn\bfX\Lie_Y\Theta_\mathcal{L}+\Lie_Y\inn{\mathbf{X}}\left(\sigma_{\Theta_\mathcal{L}}\wedge\Theta_\mathcal{L}\right)-\inn{\mathbf{X}}\Lie_Y\left(\sigma_{\Theta_\mathcal{L}}\wedge\Theta_\mathcal{L}\right)
\\ &=
\Lie_Y\inn{\mathbf{X}}\d\Theta_\mathcal{L}+\Lie_Y\inn{\mathbf{X}}\left(\sigma_{\Theta_\mathcal{L}}\wedge\Theta_\mathcal{L}\right)-\inn{\mathbf{X}}\left(\Lie_Y\sigma_{\Theta_\mathcal{L}}\wedge\Theta_\mathcal{L}\right)
\\ &=\Lie_Y\inn{\mathbf{X}}\bd\Theta_\mathcal{L}-\inn{\mathbf{X}}\left(\Lie_Y\sigma_{\Theta_\mathcal{L}}\wedge\Theta_\mathcal{L}\right)=
-\inn{\mathbf{X}}\left(\Lie_Y\sigma_{\Theta_\mathcal{L}}\wedge\Theta_\mathcal{L}\right)=0\,,
\end{align*}
where the last equality follows because $\Lie_Y\sigma_{\Theta_\mathcal{L}}=0$, which holds by Proposition \ref{Prop:Rsigma}.

Finally, for the $\overline\tau$-transversality condition, recall that, if
$\bfX_\circ\in\vf_d^m(\Theta_\mathcal{L},\omega)$ satisfies that
$\iota_{\bfX_\circ}\omega=1$; then, any multivector field
$\bfX\in\{\bfX_\circ\}$ is such that $\bfX=f\bfX_\circ$,
with $f$ a nonvanishing function, and hence the transversality condition is $\inn{\bfX}\omega=f$; therefore, one has,
\[
\inn{[Y,\bfX]}\omega=
\Lie_{Y}\inn{\bfX}\omega-\inn{\bfX}\Lie_Y\omega=\Lie_{Y}f \,.
\]
Analogously, the function $f$ can be chosen so that $\Lie_Yf$ is nonvanishing. Thus, the transversality condition is also satisfied by $[Y,\bfX]$.
Therefore, $[Y,\mathbf{X}]\in \vf_d^m(\Theta_\mathcal{L},\omega)$, for every $\mathbf{X}\in\vf_d^m(\Theta_\mathcal{L},\omega)$, yields that $Y$ is an infinitesimal generalized symmetry.
\end{enumerate}
\end{proof}

In addition, we have:
\begin{proposition}
\label{sectsol}
Let $\Phi\in{\Diff}({\cal P})$ be a natural strong Noether symmetry of a multicontact Lagrangian system $({\cal P},\Theta_{\cal L},\omega)$.
Then $\Phi$ maps holonomic sections solution to the equation \eqref{sect1L} into holonomic section solutions. 

\noindent As a consequence, the same holds for infinitesimal natural strong Noether symmetries $Y\in\mathfrak{X}({\cal P})$.
\end{proposition}
\begin{proof}
Let $\bm\psi\colon M\to {\cal P}$ be a section solution to the equations \eqref{sect1L}.
If $\Phi\in{\Diff}({\cal P})$ is a natural strong Noether symmetry, it is also a restricted symmetry and, hence, it restricts to a diffeomorphism $\varphi\colon M\to M$.
Therefore, $\Phi$ preserves the fibration $\tau\colon {\cal P}\to M$ and, hence, $\Phi\circ\bm\psi\circ\varphi^{-1}$
is another section of $\tau$. Then, on the one hand, 
as $\Phi^*\Theta_{\mathcal{L}}=\Theta_{\mathcal{L}}$, we have:
$$
(\Phi\circ\bm\psi\circ\varphi^{-1})^*\Theta_{\mathcal{L}}=
(\varphi^{-1})^*[\bm\psi^*(\Phi^*\Theta_{\mathcal{L}})]=
(\varphi^{-1})^*(\bm\psi^*\Theta_{\mathcal{L}})=0 \,.
$$
On the other hand, for every $X\in\vf({\cal P})$, 
\begin{align}
(\Phi\circ\bm\psi\circ\varphi^{-1})^*[\inn{X}\bd\Theta_{\mathcal{L}}]&=
(\bm\psi\circ\varphi^{-1})^*[(\Phi^*\inn{X}\bd\Theta_{\mathcal{L}})]=
(\bm\psi\circ\varphi^{-1})^*[\inn{\Phi_*^{-1}X}(\Phi^*\bd\Theta_{\mathcal{L}})]
\\ &=
(\bm\psi\circ\varphi^{-1})^*[\inn{\Phi_*^{-1}X}\big(\Phi^*(\d\Theta_{\mathcal{L}}+\sigma_{\Theta_\mathcal{L}}\wedge\Theta_\mathcal{L})\big)]
\\ &=
(\bm\psi\circ\varphi^{-1})^*[\inn{\Phi_*^{-1}X}\big(\d(\Phi^*\Theta_{\mathcal{L}})+(\Phi^*\sigma_{\Theta_\mathcal{L}})\wedge(\Phi^*\Theta_\mathcal{L})\big)]
\\ &=
(\bm\psi\circ\varphi^{-1})^*[\inn{\Phi_*^{-1}X}(\d\Theta_{\mathcal{L}}+\sigma_{\Theta_\mathcal{L}}\wedge\Theta_\mathcal{L})]
\\ &=
(\varphi^{-1})^*[\bm\psi^*\inn{X'}\bd\Theta_{\mathcal{L}}]=0 \,,
\end{align}
since $\Phi^*\sigma_{\Theta_\mathcal{L}}=\sigma_{\Theta_\mathcal{L}}$,
by the second item of Lemma \ref{Prop:Rsigma}, and $X'=\Phi_*^{-1}X\in\vf({\cal P})$ is an arbitrary vector field.
Thus, $\Phi\circ\bm\psi\circ\varphi^{-1}$ is also a section solution to the field equations.

Now, if $\bm\psi$ is a holonomic section on ${\cal P}$,
then it is of the form $\bm\psi=(j^1\phi,s)$,
where $j^1\phi\colon M\to J^1\pi$ is a holonomic section on $J^1\pi$ (see the beginning of Section \ref{sect:mcs}).
Moreover, $\Phi$ is a natural symmetry, then it is the canonical lift to ${\cal P}$ of a diffeomorphism on $\mathscr{E}$, and then, it is of the form
$\Phi=(j^1\Phi_E,\Psi)$, where $j^1\Phi_E\in{\Diff}(J^1\pi)$ and $\Psi\in\Diff(\bigwedge^{m-1}\T^*M)$.
Therefore,
$$
\Phi\circ\bm\psi\circ\varphi^{-1}=
(j^1\Phi_E\circ\ j^1\phi\circ\varphi^{-1},s\circ\varphi^{-1})=
(j^1(\Phi_E\circ\phi)\circ\varphi^{-1},s\circ\varphi^{-1})\,,
$$
and $j^1(\Phi_E\circ\phi)\circ\varphi^{-1}$ is a holonomic section on $J^1\pi$ (see \cite{Sau_89}). Hence, by definition, $\Phi\circ\bm\psi$ is also a holonomic section on ${\cal P}$.
\end{proof}

Taking into account Remark \ref{nongenerator}, another immediate consequence of Lemma \ref{Prop:Rsigma} 
is the analogous to Proposition \ref{genconsforms} for dissipative forms and strong Noether symmetries:

\begin{proposition}
\label{Nsym-disform}
Let $\bm\xi\in\df^{m-1}({\cal P})$ be a dissipative form of a multicontact Lagrangian system
$({\cal P},\Theta_{\cal L},\omega)$. Therefore:
\begin{enumerate}
\item 
If $\Phi\in{\Diff}({\cal P})$ is a strong Noether symmetry, then $\Phi^*\bm\xi$ is also a dissipative form.
\item
If $Y\in\vf({\cal P})$ is an infinitesimal strong Noether symmetry, then $\Lie_Y\bm\xi$ is also a dissipative form.
\end{enumerate} 
\end{proposition}

Finally, focusing on infinitesimal symmetries,
we can state the fundamental result that relates Noether symmetries with dissipated forms:

\begin{theorem}[Noether's Theorem]
    Let $Y\in\mathfrak{X}({\cal P})$ be an infinitesimal Noether symmetry of a multicontact Lagrangian system $({\cal P},\Theta_{\cal L},\omega)$, and let ${\bm \xi}_Y:=\iota_Y\Theta_{\cal L}$. Then, for every $\mathbf{X}\in \vf_d^m(\Theta_{\cal L})$, one has,
    \[
    \iota_{\mathbf{X}}\bd{\bm \xi}_Y=0\,.
    \]
    In other words, ${\bm \xi}_{Y}\in \Omega^{m-1}({\cal P})$ is a dissipative form associated with the infinitesimal Noether symmetry $Y\in \mathfrak{X}({\cal P})$.
\end{theorem}
\begin{proof}
    Recall that $\Lie_Y\Theta_{\cal L}=0$. Then,
\begin{align}
\iota_\bfX \bd\bm\xi_Y&=\iota_\bfX \bd \iota_Y\Theta_\mathcal{L}=\iota_\bfX \left(\d\iota_Y\Theta_\mathcal{L}+\sigma_{\Theta_\mathcal{L}}\wedge\iota_Y\Theta_\mathcal{L}\right)=-\iota_\bfX \iota_Y\d\Theta_\mathcal{L}+\iota_\bfX (\sigma_{\Theta_\mathcal{L}}\wedge\iota_Y\Theta_\mathcal{L})
\\
&=(-1)^{m+1}\iota_Y\iota_\bfX \d\Theta_\mathcal{L} + \iota_\bfX \left(-\iota_Y(\sigma_{\Theta_\mathcal{L}}\wedge\Theta_\mathcal{L})+(\iota_Y\sigma_{\Theta_\mathcal{L}})\Theta_\mathcal{L}\right)
\\
&=(-1)^{m+1}\iota_Y\iota_\bfX \bd\Theta_\mathcal{L}+(\iota_Y\sigma_{\Theta_\mathcal{L}})\iota_\bfX \Theta_\mathcal{L}=0\,
\end{align}
and the statement follows.
\end{proof}

\begin{remark}
It is interesting to point out that the only requirement that an infinitesimal symmetry is Noether (condition (i) in Definition \ref{Noethersym}) is needed to prove Noether's Theorem.
Nevertheless, the Noether symmetry must be strong (conditions (i) and (ii))
in order to prove Lemmas \ref{Reebinv} and \ref{Prop:Rsigma} 
and hence Propositions \ref{N-gen} (the relation with generalized symmetries) and \ref{Nsym-disform}.
Finally, the conditions of being natural (and then restricted) are also needed to prove Proposition \ref{sectsol}
(they transform holonomic section solutions into others).
\end{remark}

\subsection{Conservation and dissipation laws and symmetries for multicontact Hamiltonian systems}

First, to introduce {\sl natural symmetries} for multicontact Hamiltonian field theories, we need to define previously canonical lifts to ${\cal P}^*$.
Then, as at the end of Section \ref{sect:gensym},
consider the manifold $\mathscr{E}\equiv E\times_M\bigwedge^{m-1}(\Tan^*M)$, and the diagram
\begin{equation}
\xymatrix{
& & & \bigwedge^{m-1}(\cT M) \ar[dr]^{\kappa} \\
\mathcal{P}^* \ar[rr]^{\overline{\rho}} \ar@/^.5pc/[urrr]^{\overline{\tau}_2} \ar@/_.5pc/[drrr]_{} & & \mathscr{E} \ar[rr]^{\rho_\circ} \ar[ur]_{\rho_1} \ar[dr]^{\rho_2} & & M \\
& & & E \ar[ur]_{\pi}
}\nonumber
\end{equation}
Let $\Phi_\mathscr{E}\colon\mathscr{E}\to~\mathscr{E}$
be a diffeomorphism such that it restricts to a diffeomorphism $\Psi\colon\bigwedge^{m-1}(\Tan^*M)\to\bigwedge^{m-1}(\Tan^*M)$; that is, $\rho_1\circ\Phi_\mathscr{E}=\Psi\circ\rho_1$;
what means that $\Phi_\mathscr{E}$ preserves the fibration of $\rho_1$.
Let  $\Phi_E\colon E\to E$ be the restriction of $\Phi$ to the fibers
of $\rho_1$, which are identified with $E$,
and write $\Phi_\mathscr{E}=(\Phi_E,\Psi)$.
In the same way, let $Y_\mathscr{E}\in\vf(\mathscr{E})$ be a
$\rho_1$-projectable vector field.

\begin{definition}
\begin{enumerate}
\item 
The \dfn{canonical lift of a diffeomorphism} $\Phi_\mathscr{E}$ to ${\cal P}^*$
is the diffeomorphism $\overline\Phi\colon{\cal P}^*\to
{\cal P}^*$ such that $\overline\Phi:=(\overline\Phi_E,\Psi)$, where $\overline\Phi_E\colon J^1\pi^*\to J^1\pi^*$ is the canonical lift of $\Phi_E$ to $J^1\pi^*$ (see \cite{GR_24}).
\item
The \dfn{canonical lift of a vector field} to ${\cal P}^*$ of a $\rho_2$-projectable vector field $Y_\mathscr{E}\in\vf(\mathscr{E})$
is the vector field $\overline Y\in\vf({\cal P}^*)$ whose local diffeomorphisms generated by the flow of $\overline Y$ are the canonical lifts of the local diffeomorphisms generated by the flow of $Y_\mathscr{E}$ (see \cite{GR_24}).
\end{enumerate}
\end{definition}

In particular, in coordinates, starting from \eqref{vfY} we obtain
$$
\overline Y=f^\mu\parder{}{x^\mu}+F^a\frac{\partial}{\partial y^a}+\left(\parder{f^\mu}{x^\nu} p_a^\nu-\parder{f^\nu}{x^\nu} p_a^\mu-\parder{F^b}{y^a}p_b^\mu\right)\frac{\partial}{\partial p_a^\mu}
+g^\mu\parder{}{s^\mu} \,.
$$
If $\overline{\rho}\colon{\cal P}^*\to\mathscr{E}$
is the natural projection, then
$\overline{\rho}\circ\overline\Phi=\Phi_\mathscr{E}\circ\overline{\rho}$ and
$\overline{\rho}_*(\overline Y)=Y_\mathscr{E}$.

All definitions of conserved and dissipative forms, and symmetries;
as well as all the results previously obtained in the preceding sections for multicontact Lagrangian systems (including {\sl Noether's theorem\/}),
are established and proven in the same manner for multicontact Hamiltonian systems;
except for those that specifically refer to the holonomy condition of the solutions of the field equations.
In particular, this is the case of Proposition \ref{sectsol} which now
refers to sections $\bm\psi\colon M\to\mathcal{P}^*$,
and holds for restricted strong Noether symmetries, not necessarily natural.

\section{Application: a vibrating string with a linear damping}
\label{applications}

The multicontact formulation of this system has been analyzed in detail in \cite{LGMRR_23, LGMRR_25}. Consider ${\cal P}=\R^2\times\oplus^2 \T \R\times\R^2$ with coordinates $\{t,x,y,y_t,y_x,s^t,s^x\}$ and the Lagrangian describing a vibrating string with linear damping, of the form
\[
L(t,x,y,y_t,y_x,s^t,s^x)=\frac{1}{2}(\rho y_t^2 -\tau y_x^2)-\gamma s^t\, ,
\]
where $\rho,\tau,\gamma\in \R$. For holonomic sections
$\displaystyle\bm{\psi}(x^\nu)=\Big(x^\mu,y^a(x^\nu),\parder{y^a}{x^\mu}(x^\nu),s^\mu(x^\nu)\Big)$,
it leads to the Herglotz--Euler--Lagrange equation \eqref{ELeqs1} which is
\begin{equation}
\label{eq:vibrating-string-dampings}
\parder{^2y}{t^2}-\frac{\tau}{\rho}\parder{^2y}{x^2}=-\gamma\parder{y}{t} \,.
\end{equation}
Then, using \eqref{thetacoor1}, it follows that
\begin{align*}
\Theta_{\cal L} &= -\rho y_t\d y\wedge \d x-\tau y_x\d y\wedge \d t+\left(\rho y_t^2-\tau y_x^2-\frac{1}{2}(\rho y_t^2 -\tau y_x^2)+\gamma s^t{}\right)\d t\wedge \d x\\
&\quad + \d s^t\wedge \d x-\d s^x\wedge \d t
\\
&=
-\rho y_t\d y\wedge \d x-\tau y_x\d y\wedge \d t +\left(\frac{1}{2}\rho y^2_t-\frac{1}{2}\tau y_x^2+\gamma s^t\right)\d t \wedge \d x+\d s^t\wedge \d x- \d s^x\wedge \d t\,,
\end{align*}
and the energy Lagrangian function is
$\displaystyle E_\mathcal{L}=\frac{1}{2}\rho y^2_t-\frac{1}{2}y_x^2+\gamma s^t$.
Moreover,
\[
\omega=\d t\wedge \d x\,.
\]
Consider the infinitesimal symmetry $Y=\tparder{}{y}$. It is clear that the vector field $Y$ is an infinitesimal strong Noether symmetry, since it leaves $\Theta_{\cal L}\in\Omega^2(\cal P)$ and $\omega\in\Omega^2(\cal P)$ invariant. Then,
\[
{\bm \xi}_Y:=\inn{Y}\Theta_{\cal L}=-\rho y_t\d x-\tau y_x\d t\,.
\]
Furthermore, by \eqref{sigmaL}, one gets
\[
\sigma_{\cal L}=-\gamma \d t\,,
\]
and the infinitesimal strong Noether symmetry $Y$ also leaves $\sigma_{\cal L}$ invariant, as stated in Lemma \ref{Prop:Rsigma}. 
To see that $\boldsymbol{\xi}_Y\in\Omega^1(\cal P)$ is a dissipative form, one has to check that $\iota_{\mathbf{X}}\bd\boldsymbol{\xi}_Y=0$, for every $\mathbf{X}\in\mathfrak{X}^2(\Theta_{\cal L})$. First,
\begin{equation}
\label{disformweq}
\bd {\bm\xi_Y}=\d {\bm\xi_Y}+\sigma \wedge {\bm\xi_Y}=\rho \d x\wedge \d y_t+\tau \d t\wedge \d y_x+\rho\gamma y_t\d x\wedge \d t\,.  
\end{equation}
To find $\mathbf{X}$, one has that
\begin{align}
\sigma_{\cal L}\wedge\Theta_{\cal L}&=-\gamma\rho y_t\d t\wedge \d y\wedge \d x +\gamma \d t\wedge \d s^t\wedge \d x\,,
\\
\d \Theta_{\cal L}&=-\rho \d y_t\wedge \d y\wedge \d x -\tau \d y_x\wedge\d y\wedge \d t+\left(\rho y_t\d y_t-\tau y_x\d y_x+\gamma \d s^t\right)\wedge \d t\wedge \d x\,.
\end{align}
Therefore,
\begin{align*}
\bd\Theta_{\cal L} &= \d\Theta_{\cal L}+\sigma_{\cal L}\wedge \Theta_{\cal L}\\
&=\gamma\rho y_t\d t\wedge \d x\wedge \d y-\rho\d y_t\wedge \d y\wedge \d x-\tau \d y_x\wedge \d y\wedge \d t
\\ &\quad
+\rho y_t\d y_t\wedge \d t\wedge \d x -\tau y_x \d y_x\wedge \d t\wedge \d x\,.
\end{align*}
Now, $\mathbf{X}$ must be solution to the equations
\begin{equation}
\label{Ex::SolConditions}
\inn{\mathbf{X}}\Theta_{\cal L}=0\,,\qquad \inn{\mathbf{X}}\bd\Theta_{\cal L}=0\,,\qquad \inn{\mathbf{X}}\omega=1\,.
\end{equation}
Recall that $\mathbf{X}=(X_1,X_2)$, where
\begin{align*}
    X_1 &= A_1\frac{\partial}{\partial t}+A_2\frac{\partial}{\partial x}+A_3\frac{\partial}{\partial y}+A_4\frac{\partial}{\partial y_t}+A_5\frac{\partial}{\partial y_x}+A_6\frac{\partial}{\partial x^t}+A_7\frac{\partial}{\partial s^x}\,,\\
    X_2 &= B_1\frac{\partial}{\partial t}+B_2\frac{\partial}{\partial x}+B_3\frac{\partial}{\partial y}+B_4\frac{\partial}{\partial y_t}+B_5\frac{\partial}{\partial y_x}+B_6\frac{\partial}{\partial x^t}+B_7\frac{\partial}{\partial s^x}\,.
\end{align*}
Imposing equations \eqref{Ex::SolConditions}, $\mathbf{X}$ takes the form
\begin{align*}
    X_1 &= \frac{\partial}{\partial t}+y_t\frac{\partial}{\partial y}+\left(\frac{\tau}{\rho}B_5-\gamma y_t\right)\frac{\partial}{\partial y_t}+A_5\frac{\partial}{\partial y_x}+\left(L-B_7\right)\frac{\partial}{\partial s^t}+A_7\frac{\partial}{\partial s^x}\,,\\
    X_2 &= \frac{\partial}{\partial x}+y_x\frac{\partial}{\partial y}+B_4\frac{\partial}{\partial y_t}+B_5\frac{\partial}{\partial y_x}+B_6\frac{\partial}{\partial x^t}+B_7\frac{\partial}{\partial s^x}\,,
\end{align*}
where $A_5,A_7,B_4,B_5,B_6,B_7\in \Cinfty(\cal P)$.

Using the expression \eqref{disformweq}, it is immediate to check that $\inn{\mathbf{X}}\bd\boldsymbol{\xi}_Y=0$, and thus $\boldsymbol{\xi}_Y$ is a dissipated form.
To obtain the corresponding dissipation law, consider a holonomic section
$\displaystyle\bm\psi(t,x)=\Big(t,x,y(t,x),\parder{y}{t}(t,x),\parder{y}{x}(t,x),s^t(t,x),s^x(t,x)\Big)$
solution to the field equations \eqref{sect1L}.
Then, following the procedure explained in Remark 
\ref{remdislaw}, first we have
$$
\bm\psi^*\bm\xi_Y=-\rho\parder{y}{t}\d x-\tau\parder{y}{x}\d t\,,
$$
and therefore
$$
X_{\bm\psi^*\bm\xi_Y}=-\rho\parder{y}{t}\parder{}{t}+\tau\parder{y}{x}\parder{}{x} \,,
$$
and the corresponding dissipation law \eqref{dislaweq} is the field equation \eqref{eq:vibrating-string-dampings} itself.

For the Hamiltonian formalism,
${\cal P}^*= \R^2\times\bigoplus^2\T^*\R\times\R^2$,
with coordinates $\{t,x,y,p^x,p^t,s^t,s^x\}$.
The Legendre map reads
$$
\mathcal{FL}(t,x,y,y_t,y_x,s^t,s^x)=
(t,x,y,p^t=\rho y_t,p^x=-\tau y_x,s^t,s^x) \,,
$$
and it is a diffeomorphism. Then, the Hamiltonian function of the vibrating string with a linear dumping reads
\[H(t,x,y,p^t,p^x,s^t,s^x)=(\mathcal{FL}^{-1})^*E_\mathcal{L}=\frac{1}{2\rho}(p^t)^2-\frac{1}{2\tau}(p^x)^2-\gamma s^t\,.\]
and
$$
\Theta_{\cal H}=
-\rho y_t\d y\wedge \d x-\tau y_x\d y\wedge \d t +\left(\frac{(p^t)^2}{2}-\frac{(p^x)^2}{2}+\gamma s^t\right)\d t \wedge \d x+\d s^t\wedge \d x- \d s^x\wedge \d t\,.
$$
The expression of the Noether symmetry is the same, and everything develops analogously to the Lagrangian case.

\section{Conclusions and outlook}
\label{Sec:Conclusions}

In this work, we have extended the fundamental concepts of symmetries, conservation laws, and Noether's Theorem from the classical multisymplectic field theory to the setting of action-dependent field theories within the multicontact framework. By formulating both the Lagrangian and Hamiltonian descriptions in this context, we have showed how symmetries naturally give rise to dissipation laws, thereby redefining the role of conserved and dissipated quantities in this extended framework.  

A key contribution of our study is the characterization of symmetries associated with both the field equations and the multicontact structure. This led to the introduction of strong Noether symmetries, which provide a direct connection between symmetries and dissipated quantities. The establishment of a generalized Noether's Theorem in this setting reinforces the fundamental link between these symmetries and the underlying physical and geometric structures governing action-dependent field theories.  

While our results offer a comprehensive extension of Noether's Theorem in the multicontact framework, several open questions remain. Future research could further explore the implications of strong Noether symmetries in specific physical models, as well as their role in more general geometric settings. Additional directions of study include the analysis of equivalent Lagrangians for multicontact field theories and their associated symmetries, as well as the extension of Noether's Theorem in this context. Furthermore, investigating symmetries, conservation, and dissipation laws in singular (premulticontact) Lagrangian and Hamiltonian field theories, including the study of geometric gauge symmetries, is an important avenue for future work. Lastly, advances towards reduction theorems in this framework could provide deeper insights into the structure and applications of multicontact field theories.

\appendix

\section{Multivector fields on manifolds and fiber bundles}\label{append}

This appendix provides a review on multivector fields (see, for instance, \cite{BCGGG_91,CIL_96,EMR_98} for more details).

Let $\mathcal{M}$ be a $N$-dimensional differentiable manifold.
The \dfn{$m$-multivector fields}
 in $\mathcal{M}$ ($m\leq N$)
are the skew-symmetric contravariant tensor fields of order $m$ in $\mathcal{M}$
or, what is the same thing, sections of the $m$-multitangent bundle
$\bigwedge^m\Tan\mathcal{M}:=\Tan\mathcal{M}\wedge\overset{m}{\dotsb}\wedge\Tan\mathcal{M}$.
The set of them is denoted $\vf^m (\mathcal{M})$.

If $\mathbf{X}\in\mathfrak{X}^m(\mathcal{M})$,
for every point $\mathrm{p}\in \mathcal{M}$,
there is an open neighbourhood $U_\mathrm{p}\subset \mathcal{M}$ and
$X_1,\ldots ,X_r\in\mathfrak{X}(U_\mathrm{p})$ such that,
for $m \leq r\leq\dim\mathcal{M}$,
$$
\mathbf{X}\vert_{U_{\mathrm{p}}}=\sum_{1\leq i_1<\ldots <i_m\leq r} f^{i_1\ldots i_m}X_{i_1}\wedge\dotsb\wedge X_{i_m} \, ,
$$
with $f^{i_1\ldots i_m} \in \Cinfty(U_\mathrm{p})$.
In particular, $\mathbf{X}\in\vf^m(\mathcal{M})$ is said to be a
\dfn{locally decomposable multivector field} if
there exist $X_1,\ldots ,X_m\in\vf(U_\mathrm{p})$ such that $\mathbf{X}\vert_{U_\mathrm{p}}=X_1\wedge\dotsb\wedge X_m$.

Given a decomposable multivector field $\bfX = X_1\wedge\dotsb\wedge X_m$ with local expression
$$ X_\alpha = X_\alpha^i\parder{}{x^i}\,, $$
we say that a map $\psi\colon\R^m\to {\cal M}$ is an \dfn{integral map} of $\bfX$ if it satisfies the set of partial differential equations
\begin{equation}\label{eq:pde-multi}
    \parder{\psi^i}{t^\alpha} = X_\alpha^i\circ\psi\,.
\end{equation}

Locally decomposable $m$-multivector fields are locally associated with $m$-dimensional
distributions $D\subset\Tan \mathcal{M}$.
This splits the set of locally decomposable multivector fields into {\sl equivalence classes}, $\{\bfX\}\subset\vf^m_d(\mathcal{M})$ 
which are made of the locally decomposable multivector fields associated with the same distribution.
If $\bfX ,\bfX '\in\{ \bfX \}$ then, for $U\subset \mathcal{M}$,
there exists a non-vanishing function $f\in\Cinfty(U)$ such that 
$\bfX '=f\bfX $ on $U$.
In addition, an \dfn{integrable multivector field}
is a locally decomposable multivector field whose associated distribution is integrable; that is, involutive. Suppose we want the multivector field to have integral maps. In that case, we need to impose the stronger integrability condition $[X_\alpha, X_\beta] = 0$ for every $\alpha,\beta = 1,\dotsc,m$, which is precisely the integrability condition of the system of PDEs \eqref{eq:pde-multi} \cite{Lee_12}.

In particular, let $\varrho\colon\mathcal{M}\to M$ be a fiber bundle.
A multivector field $\mathbf{X}\in\mathfrak{X}^m(\mathcal{M})$ is \emph{$\varrho$-transverse} if,
for every $\beta\in\Omega^m(M)$ such that $\beta_{\varrho(\mathrm{p})}\neq 0$,
at every point $\mathrm{p}\in \mathcal{M}$, we have that
$(\inn{\bfX}(\varrho^*\beta))_{\mathrm{p}}\neq 0$.
Then, if $\mathbf{X}\in\mathfrak{X}^m(\mathcal{M})$ is integrable and $\varrho$-transverse, 
its integral manifolds are local sections of the projection $\varrho$.

If $\alpha\in\df^k(\mathcal{M})$ and $\mathbf{X}\in\mathfrak{X}^m(\mathcal{M})$,
the \dfn{contraction} between $\bfX $ and $\alpha$ is the natural contraction between tensor fields; in particular, it is zero when $k<m$ and, if $k\geq m$,
$$
 \inn{\bfX}\alpha\mid_{U}:= \sum_{1\leq i_1<\dotsb <i_m\leq
 r}f^{i_1\dotsb i_m} \inn{X_1\wedge\dotsb\wedge X_m}\alpha 
=
 \sum_{1\leq i_1<\dotsb <i_m\leq r}f^{i_1\ldots i_m} \inn{X_m}\dotsb\inn{X_1}\alpha \,.
$$

The \dfn{Lie derivative} of $\alpha$ with respect to $\bfX $ is the graded bracket (of degree $m-1$)
$$
\Lie_\bfX \alpha:=[\d , \inn{\bfX}]\alpha=
(\d\inn{\bfX} - (-1)^m\inn{\bfX}\d)\alpha \,.
$$

If $\bfX\in\vf^m(\mathcal{M})$ and $\bfY\in\vf^n(\mathcal{M})$,
the \dfn{Schouten--Nijenhuis bracket} of $\bfX ,{\bf Y}$  \cite{KMS_93,Mar_97} is
the graded commutator of $\Lie_\bfX$ and $\Lie_\bfY$; that is,
the multivector field $[\bfX,\bfY]\in\vf^{m+n-1}(\mathcal{M})$, such that,
\begin{equation}
\label{liebrac}
\Lie_{[\bfX,\bfY]}=[\Lie_\bfX,\Lie_\bfY]=\Lie_\bfX\Lie_\bfY-\Lie_\bfY\Lie_\bfX \,.
\end{equation}
For (locally) decomposable multivector fields, it is defined as follows: 
given $\bfX=X_1\wedge\cdots\wedge X_m$ and $\bfY=Y_1\wedge\cdots\wedge Y_n$, their Schouten--Nijenhuis bracket is
\begin{equation}
[\bfX,\bfY]:=\sum_{i,j}(-1)^{i+j}[X_i,Y_j]\wedge 
X_1\wedge\cdots\wedge X_{i-1}\wedge X_{i+1}\wedge\cdots\wedge X_m\wedge Y_1\wedge\cdots\wedge Y_{j-1}\wedge Y_{j+1}\wedge\cdots\wedge Y_n \,,    
\end{equation}
and it extends linearly to any multivector fields in general. Let us recall that the following properties hold \cite[Proposition  3.1]{Mar_97}:
given $\mathbf{P}\in \mathfrak{X}^m(M)$, $\mathbf{Q}\in\mathfrak{X}^n(M)$, and $R\in\mathfrak{X}(M)$, then,
\begin{equation}
[\mathbf{P},\mathbf{Q}]=-(-1)^{(m-1)(n-1)}[\mathbf{Q},\mathbf{P}]\,,\label{MarleMultiDer1}
\end{equation}
and
\begin{equation}
    [\mathbf{P},\mathbf{Q}\wedge R]=[\mathbf{P},\mathbf{Q}]\wedge R+(-1)^{(m-1)n}\mathbf{Q}\wedge[\mathbf{P},R]\,.         
\label{MarleMultiDer2}
\end{equation}
Another useful property is the following:
\begin{proposition}
\label{Prop::A1}
    Let $\mathbf{X}\in \mathfrak{X}^k(M)$ and $Y\in\mathfrak{X}(M)$. Then,
    \[
    \iota_{[Y,\mathbf{X}]}=\Lie_Y\iota_{\mathbf{X}}-\iota_{\mathbf{X}}\Lie_Y\,.
    \]
\end{proposition}
\begin{proof}
Note that the statement holds if $k=1$.
For locally decomposable multivector fields \(\mathbf{X}=X_1\wedge\ldots\wedge X_k\), 
the proof is performed using induction with respect to $k$. 
Additionally, let $\widetilde{\mathbf{X}}=X_1\wedge\ldots\wedge X_k\wedge X_{k+1}$. 
Then, using \eqref{MarleMultiDer1} and \eqref{MarleMultiDer2}, it follows that
\begin{align}
\iota_{[Y,\widetilde{\mathbf{X}}]}&=
\iota_{[Y,\mathbf{X}\wedge X_{k+1}]}=\iota_{[Y,\mathbf{X}]\wedge X_{k+1}}+\iota_{\mathbf{X}\wedge [Y,X_{k+1}]}=\iota_{[Y,\mathbf{X}]}\iota_{ X_{k+1}}+\iota_{\mathbf{X}}\iota_{[Y,X_{k+1}]}
\\ &=\left(\Lie_Y\iota_{\mathbf{X}}-\iota_{\mathbf{X}}\Lie_Y\right)\iota_{X_{k+1}}+\iota_{\mathbf{X}}\left(\Lie_Y\iota_{X_{k+1}}-\iota_{X_{k+1}}\Lie_Y\right)=\Lie_Y\iota_{\widetilde{\mathbf{X}}}-\iota_{\widetilde{\mathbf{X}}}\Lie_Y\,,
\end{align}
and the statement follows, since it can be locally extended linearly to any multivector fields in general, not only decomposable ones.
\end{proof}

\addcontentsline{toc}{section}{Acknowledgements}
\section*{Acknowledgements}

We acknowledge the financial support of the 
{\sl Ministerio de Ciencia, Innovaci\'on y Universidades} (Spain), projects PID2021-125515NB-C21, and RED2022-134301-T of AEI,
and Ministry of Research and Universities of
the Catalan Government, project 2021 SGR 00603 \textsl{Geometry of Manifolds and Applications, GEOMVAP}. B.M. Zawora acknowledges funding from the IDUB mikrogrant program to accomplish a research stay at the UPC in Barcelona.

\bibliographystyle{abbrv}
{\small
\bibliography{references.bib}
}

\end{document}